\documentclass[a4paper,twocolumn,11pt,accepted=2025-09-12]{quantumarticle}
\pdfoutput=1
\usepackage[utf8]{inputenc}
\usepackage[english]{babel}
\usepackage[T1]{fontenc}
\usepackage[colorlinks]{hyperref}
\providecommand{\doi}[1]{\href{https://doi.org/#1}{https://doi.org/#1}}
\usepackage{amsmath,amsfonts,amssymb}
\usepackage{xcolor}
\usepackage{graphicx}
\usepackage{float}
\usepackage{amsthm}
\usepackage{appendix}
\usepackage{blkarray}
\usepackage{multirow}
\usepackage{changes}
\usepackage{comment}
\usepackage[ruled,vlined]{algorithm2e}
\usepackage{algorithmic}
\usepackage{braket}
\usepackage[numbers, sort&compress]{natbib}
\usepackage[qm]{qcircuit}

\def\CH{{\cal H}}

\theoremstyle{definition}
\newtheorem{definition}{Definition}[section]
\newtheorem{theorem}{Theorem}[section]
\newtheorem{lemma}{Lemma}[section]
\newtheorem{proposition}{Proposition}[section]
\newtheorem{remark}{Remark}[section]
\newtheorem{corollary}{Corollary}[section]

\numberwithin{equation}{section}

\newcommand{\CC}[1]{\textcolor{black}{(CC) #1}}

\DeclareMathOperator{\Tr}{Tr}

\newcommand{\RS}[1]{\textcolor{red}{(RS) #1}}

\begin{document}

\title{Quantum Lego and XP Stabilizer Codes}
\author{Ruohan Shen}
 \email{rhshen@mit.edu}
  \affiliation{Department of Physics, Tsinghua University, Beijing 100084, China}
 \affiliation{Institute for Quantum Information and Matter
  California Institute of Technology,
  1200 E California Blvd, Pasadena, CA 91125, USA.}

\author{Yixu Wang}
 \email{wangyixu@mail.tsinghua.edu.cn}
 \affiliation{Institute for Advanced Study, Tsinghua University, Beijing 100084, China}
\author{ChunJun Cao}
 \email{cjcao@vt.edu}
 \affiliation{Institute for Quantum Information and Matter
  California Institute of Technology,
  1200 E California Blvd, Pasadena, CA 91125, USA.}
 \affiliation{Department of Physics, Virginia Tech, Blacksburg, VA, 24061, USA}

\begin{abstract}
We apply the recent graphical framework of ``Quantum Lego'' to XP stabilizer codes where the stabilizer group is generally non-Abelian. We show that the idea of operator matching continues to hold for such codes and is sufficient for generating all their XP symmetries provided the resulting code is XP. We provide an efficient classical algorithm for tracking these symmetries under tensor contraction or conjoining. This constitutes a partial extension of the algorithm implied by the Gottesman-Knill theorem beyond Pauli stabilizer states and Clifford operations. Because conjoining transformations generate quantum operations that are universal, the XP symmetries obtained from these algorithms do not uniquely identify the resulting tensors in general. Using this extended framework, we provide examples of novel XP stabilizer codes with a higher distance than existing non-trivial XP regular codes and a $[[8,1,2]]$ Pauli stabilizer code with a fault-tolerant $T$ gate. For XP regular codes,  we also construct a tensor-network-based maximum likelihood decoder for any independently and identically distributed single qubit error channel using weight enumerators.
\end{abstract}

\maketitle

\section{Introduction}
The stabilizer formalism~\cite{Gottesman1997} or, more precisely, the Pauli stabilizer formalism (PSF), has made a profound impact in numerous areas of quantum computation and communication, such as quantum error correction, simulation, and quantum networks. While powerful in providing efficient classical simulations of stabilizer states and Clifford processes, they are limited in describing the full range of operations that are necessary for quantum advantage. For example, it is known that novel quantum error correcting codes (QECCs), beyond those described by the Pauli stabilizer formalism, can be the key to improving universal fault-tolerant quantum computation, in constructing stable self-correcting quantum memories, and in understanding quantum many-body systems with non-Abelian symmetries, e.g.~\cite{Kitaev:1997wr,wootton,wenlevin,QMreview}. Therefore, explorations beyond PSF are pivotal as we move toward fault-tolerance.

Several notable extensions beyond PSF build on its success by allowing the underlying stabilizer subgroup to be non-Abelian~\cite{Ni2015,Webster2022}. Among them is the XP formalism (XPF)~\cite{Webster2022}, which contains the PSF as a special case; and in general considers stabilizer states and codes defined by the non-Abelian subgroups of the XP group. In particular, pioneering work by Webster et al.~\cite{Webster2022} discusses the construction of XP stabilizer codes, their connection with weighted hypergraph states and their classical simulability. Furthermore, the formalism provides a more natural language to identify transversal non-Clifford gates~\cite{Webster2023} which can be useful for the design of novel magic state distillation protocols. However, much remains to be explored for the XP stabilizer codes compared to its Pauli stabilizer counterparts. This includes many key properties of XP stabilizer states as well as an efficient description of their transformation under various quantum processes. Before one can convincingly demonstrate any advantage of XP codes for fault-tolerance, we also need to lay down some necessary groundwork, such as a more comprehensive characterization of the code properties, an explicit construction of XP decoders, and a larger collection of non-trivial XP stabilizer codes that are distinct from Pauli stabilizer codes. Like all quantum codes, it is also necessary to improve its scalability and flexibility when engineering larger codes with more desirable error correction properties. 

In this work, we tackle the aforementioned challenges by applying the nascent Quantum Lego (QL) formalism~\cite{Cao2022} to XP codes. In doing so, we rephrase XP regular codes as Quantum Lego blocks and provide novel insights concerning XP codes. {\color{black}XP regular codes are XP stabilizer codes that are structurally similar to CSS codes \cite{Webster2022} and encode qubits in qubits, as opposed to non-regular codes which are similar to codeword stabilized quantum codes~\cite{Cross_2009} and can encode more general structures on qubit degrees of freedom.} We then discuss how these XP Legos can be combined to create larger codes and especially how they transform under tensor contractions. Writing these codes in a check matrix representation, we provide an efficient algorithm to track its XP symmetries under operator matching, or equivalently, how the state/code transforms under ``Lego fusion'' or conjoining. \textcolor{black}{XP symmetry is a subgroup of the XP group that leaves the encoding isometry invariant. By keeping track of XP symmetries, it is efficient to describe XP codes under various operations.} Because such conjoining operations over XP Legos can generate universal gate sets and measurements, they further induce an algorithm for describing how XP states transform under quantum processes allowed in the XPF, which include the Clifford operations and those from the higher Clifford hierarchy. However, this classically efficient description is not without limitations ---  classical algorithms cannot describe universal quantum processes in polynomial time unless $BQP\subset P$. Indeed, the contraction of XP Legos does not remain an XP stabilizer state in general --- such efficient representations over finite rings are only guaranteed to fix the output state uniquely if it is also XP. When the combined state is not XP, then operator matching identifies a set of states with those matching symmetries.

By leveraging the recent developments in quantum tensor weight enumerators~\cite{Cao2022_2,Cao2023}, we also provide a method to compute the enumerator for any XP code constructed from QL, which contains key properties of the code, such as its (biased) distance.  In particular, we provide an explicit maximum-likelihood (ML) decoder for XP regular codes under any independently and identically distributed (i.i.d.) single site error channel. Because any XP code can be built from QL with the appropriate choice of Legos, this constitutes a tensor-network-based recipe for constructing an ML decoder for the class of XP regular codes. 

Finally, note that the general problem of constructing interesting XP codes can be difficult, partly due to the prohibitive cost of the search. Using QL, we drastically reduce this cost by generating novel codes from fusing XP Lego blocks and providing some working examples under this framework. Note that QL is also powerful in designing codes with target transversal single qubit gates starting from suitable Lego blocks. As such, we produce a $[[7,1,1]]$ code with transversal logical $T$ that can be converted into a $[[8,1,2]]$ code with fault-tolerant logical $T$. The latter constitutes the smallest example with such a fault-tolerant gate to the best of our knowledge and may be used for magic state distillation. In addition, although QL can accommodate any quantum codes in principle, all known QL examples are constructions involving Pauli stabilizer codes. Hence our work also provides the first instance of creating non-(Pauli)-stabilizer codes from QL.

In Sec.~\ref{Section2}, we review the basics of the XPF and discuss some novel results that we will use for this work. In Sec.~\ref{Section3}, we discuss the conjoining of XP Legos and their properties. In Sec.~\ref{Section4} we construct optimal decoders for regular XP codes and discuss how enumerators may be used to analyze their properties. This includes, for example, an optimal decoder for the double semion model~\cite{Levin2005} under any i.i.d. single site error.  We then provide examples of XP codes that have not been previously discovered in Sec.~\ref{Section5}, including a novel $[[7,1,3]]$ XP code that is not local unitary equivalent to the Steane code and a Pauli stabilizer $[[8,1,2]]$ code with fault-tolerant logical $T$ gate. Finally, we summarize our findings in Sec.~\ref{Section6} and conclude with a few forward-looking remarks. The text is structured in such a way that readers can comprehend the general idea without extensive background on XPF or QL. In the Appendices, we provide novel results, a more in-depth discussion, and proofs for theorems in the main text. 

\section{The XP Stabilizer Formalism}\label{Section2}
The XP formalism (XPF) introduced by \cite{Webster2022} revolves around the XP group,
    $\mathcal{G}_{N}=\langle \omega I,X,P\rangle$, 
 which captures a set of unitary, but not necessarily Hermitian, operators acting on a qubit. Its generators are
\begin{align}
    \omega I:&=\begin{pmatrix}
        e^{i\pi/N} & 0 \\
        0 & e^{i\pi/N} 
    \end{pmatrix},\\\nonumber
    P:&=\begin{pmatrix}
    1 & 0\\
    0 & \omega^2
    \end{pmatrix},~
    X:=\begin{pmatrix}
        0 & 1 \\
        1 & 0
    \end{pmatrix}
\end{align}
where the integer parameter $N\geq 2$ is known as the \emph{precision} of the group\footnote{In this work, we do not consider cases where $N$ is not an integer, even though such an extension should be possible.}. Structurally, the XP group for a single qubit is related to the dihedral group $D_{N}$ of order $2N$ such that $\mathcal{G}_N/\mathcal{Z}(\mathcal{G}_N)\cong D_{N}$. Here $\mathcal{Z}(\mathcal{G}_N) =\langle \omega I\rangle$ represents the center of the group. It is clear that the phase operator $P$ generates the rotation and $X$ generates reflections. 

The XP group $\mathcal{G}_{N}^{n}=\mathcal{G}_{N}^{\otimes n}$ over $n$ qubits can be obtained by taking its $n$-fold tensor product. An Abelian subgroup generated by XP strings consisting of the tensor product of $\omega I$ and $P$ is sometimes known as the \emph{diagonal subgroup} of $\mathcal{G}_N$ and its elements the \emph{diagonal operators}. Notice that for $N=2$, the XP group reduces to the Pauli group $\mathcal{P}^n$, while for $N=4$ it reduces to the XS group \cite{Ni2015}. Like the Pauli group, the XP group is clearly non-Abelian; but instead of having group elements that commute up to a global phase like the Pauli group and its qudit extensions~\cite{Gheorghiu2014}, the XP operators commute up to a diagonal operator: $PX=\omega^2 XP^{-1}=(XP)(\omega^2 P^{-2})$.

Like the Pauli group, it is also convenient to represent the elements of the XP group as a vector $\mathbf{u}=(\mathbf{x}|\mathbf{z}|p)\in  \mathbb{Z}_2^n \times \mathbb{Z}_N^n \times \mathbb{Z}_{2N}$, then the group element corresponds to $\mathbf{u}$ is
\begin{equation}
    XP_N(\mathbf{u})=\omega^p \bigotimes_{0\leq i <n} X^{\mathbf{x}[i]} P^{\mathbf{z}[i]}.
\end{equation}
For example, the so-called \emph{anti-symmetric operator} introduced by \cite{Webster2022} is
\begin{equation}
    D_N(\mathbf{z})=XP_N(\mathbf{0}|-\mathbf{z}|\sum_i \mathbf{z}[i])
\end{equation}
so the multiplication of XP group elements can be written as 
\begin{equation}
    XP_N(\mathbf{u}_1) XP_N(\mathbf{u}_2) = XP_N(\mathbf{u}_1+\mathbf{u}_2)D_N(2\mathbf{x}_2\mathbf{z}_1).
    \label{eqn:rowop}
\end{equation}
Here both the addition and multiplication are entry-wise.

Any subgroup of the XP group $\mathcal{S}\subset \mathcal{G}_N^{n}$ can be expressed as a check matrix where the generators of the group correspond to  the rows of this matrix and group multiplication is mapped to the special row operation shown in (\ref{eqn:rowop}). Similar to PSF, the three sets of columns of the check matrix respectively describe the power of the $X$ operators, the power of $P$ operators, and finally the global phase of the stabilizer. However, note also the distinction from PSF as the $\mathbf{z}$ section of the matrix no longer lives in a binary field.

Such a check matrix can be rearranged into the \emph{canonical form} where the rows correspond to the \emph{canonical generators}. These generators are either the diagonal operators $\mathbf{S}_Z=\{S_{Z_i}\}$, for which $\mathbf{x}=\mathbf{0}$, or non-diagonal operators $\mathbf{S}_X=\{S_{X_j}\}$ where $\mathbf{x}\ne \mathbf{0}$. In this form, the submatrices formed by $\mathbf{S}_X,\mathbf{S}_Z$ have the following properties:
\begin{itemize}
    \item The X-part of $\mathbf{S}_X$ is in the Reduced Row Echelon Form (RREF).
    \item The P-part of $\mathbf{S}_Z$, with the phases, is in the Howell Form, which can be seen as a generalization of RREF to rings.
\end{itemize}

Any operator $g\in \mathcal{S}$ can be expressed in the form of $g=\Pi_i{S}_{X_i}^{a_i} \Pi_j {S}_{Z_j}^{b_j}$, for some $i,j$ and $a_i\in \mathbb{Z}_2,$ $b_j \in \mathbb{Z}_N$.

With the subgroup $\mathcal{S}$, one can define a Hilbert subspace $\mathcal{C}=\{|\psi\rangle: \forall g\in \mathcal{S}, g|\psi\rangle =|\psi\rangle\}$. If $\mathcal{C}$ is not empty, then it identifies the code subspace of an XP stabilizer code. Note that the stabilizer group $\mathcal{S}$ is generally non-Abelian.
As in the PSF, the projector onto the code subspace can be represented in a compact form in terms of the stabilizer elements.
\begin{proposition}\label{prop:2.1}
    Let $\mathcal{C}$ be an XP stabilizer code with canonical generators $\mathbf{S}_X, \mathbf{S}_Z$, then projector onto $\mathcal{C}$ can be written as
\begin{equation}
    \begin{aligned}
    \Pi_{\mathcal{C}}&=\frac{1}{2^{n_X}N^{n_Z}}\sum_{\substack{i=1,...,n_X\\m_i=\{0,1\}}}\sum_{\substack{j=1,...,n_Z\\l_j=\{0,..,N-1\}}}\prod_i S_{X_i}^{m_i}\prod_{j} S_{Z_j}^{l_j}\\&=\prod_{i}^{n_X}\left(\frac 12\sum_{m_i=0}^{1}S_{X_i}^{m_i}\right)\prod_{j}^{n_Z}\left(\frac 1 N\sum_{l_j=0}^{N-1}S_{Z_j}^{l_j}\right)
    \end{aligned}
    \label{eqn:codeproj}
\end{equation}
where $n_X=|\mathbf{S}_X|, n_Z=|\mathbf{S}_Z|$.
\end{proposition}

Recall that in PSF, the stabilizer group (with phases) uniquely defines the Pauli stabilizer code and vice versa. Interestingly, this is no longer true in XPF, where generally multiple stabilizer groups can share the same code subspace $\mathcal{C}$. Consider a simple example where $\mathcal{C}=\{|0\rangle\}$ fixed by $\mathcal{S}_1=\langle Z\rangle$. However, for any even $N>2$, we realize that the group $\mathcal{S}_{\rm LID}=\langle P\rangle>\mathcal{S}_1$ also stabilizes $|0\rangle$. In fact, any non-trivial subgroup of $\mathcal{S}_{\rm LID}$ does so. Intuitively, this is because the entire XP symmetry of a code is often unnecessary in identifying $\mathcal{C}$. To resolve this ambiguity, it is often convenient to track the unique subgroup $\mathcal{S}_{\rm LID}$ that captures all of the XP symmetries of the code. This group is known as the \emph{logical identity group (LID)}, which is a group that contains all stabilizer groups of $\mathcal{C}$ as subgroups at any fixed precision.

One may also recall that an $[[n,k]]$ Pauli stabilizer code has $n-k$ stabilizer generators and wonder whether the same holds for XP codes. Although this is not true in general, we find that a version of this relation can be preserved for a subclass of XP codes known as the \emph{regular XP codes} when the precision is a power of 2. Let $\mathbf{L}_X$ be the set of  non-diagonal logical generators  of an XP-regular code.

\begin{theorem}\label{counting theorem}
    An $n$-qubit XP-regular code $\mathcal{C}$ with precision $N=2^t$ for some positive integer $t$ must have $|\mathbf{S}_X|+|\mathbf{L}_X|+|\mathbf{S}_Z|=n$.
\end{theorem}
 For such regular codes, the logical subalgebra encodes $k=|\mathbf{L}_X|$ tensor product of qubits, i.e.,  hence the number of generators of the LID group is precisely $n-k$. 
For proofs and extended discussions related to (regular) XP codes, see App.~\ref{app:xpproperty}.

\section{Combining XP Legos}\label{Section3}
\subsection{Quantum Lego}
Before combining ``Quantum Lego blocks'', let us briefly discuss what they are and review the idea behind QL. Readers familiar with this framework can skip to Sec.~\ref{subsec:xpfusion}. QL is a graphical framework for building QECCs from smaller codes that are often called atomic Legos, or seed codes in \cite{Cao2022,TNC,LTNC}. These atomic Legos are then joined together in a way that generalizes code concatenation. In particular, it is shown \cite{Cao2022} that with suitable choices of simple ``Lego blocks'', the operations in QL can generate all quantum codes, which would include the XP codes.
Another feature of QL that distinguishes it from the conventional tensor network (TN) approach or graphical calculus like \cite{ZX1,ZX2} is that it uses the symmetries of the code, which is often required for the construction of fault-tolerant gates, to track and design its properties more efficiently. Such a process is known as operator pushing or operator matching, which has demonstrated its power in controlling the check weights and recovering non-trivial constructions such as topological codes with simple and graphically intuitive reasonings. More recently, it has been discovered that the properties and decoders of codes built from QL can also be obtained up to exponentially more efficiently using a new technique based on tensor weight enumerator polynomials \cite{Cao2022_2,Cao2023}. As the rules for code generation are incredibly simple, it has also been used for automated code design with strategies based on reinforcement learning~\cite{Su2023,Mauron:2023wnl}. Most importantly for this work, QL is also valid beyond PSF. 

We can further summarize the QL operations in 3 aspects:
\begin{itemize}
    \item Identifying the atomic Legos
    \item The combination of such building blocks, and 
    \item The designation of logical vs physical degrees of freedom.
\end{itemize}

The identification of atomic Legos is a choice where we restrict ourselves to building codes from a small selection of seed codes. It is often easy to find codes over a small number of qubits with brute force search in obtaining its encoding isometry $V$. In QL, we represent $V$ graphically (Fig.~\ref{fig:qlego}a) by its coefficients, which can be written as a tensor $V_{i_1i_2\dots}$. Because the tensor is oblivious to whether it was obtained from a state or a map, we use these concepts interchangeably in QL. Graphically, the specific designation of a tensor as a map or state comes from whether we appoint a ``leg'' to be physical (outputs) or logical (inputs). Mathematically, this interchangeable designation between state and encoding map is given by the Choi-Jamiolkowski isomorphism~\cite{Choi1975,Jamiołkowski1972}. Moreover, QL is designed to track the symmetries, instead of the components, of the tensor (Fig.~\ref{fig:qlego}b). 
\begin{definition}
    A unitary operator $U$ is said to describe a symmetry of a tensor $V_{i_1, \dots, i_n}$ if for $|V\rangle = \sum_{i_1,\dots, i_n} V_{i_1,\dots, i_n}|i_1,\dots,i_n\rangle$
    $$U|V\rangle = |V\rangle.$$
\end{definition}
In other words, $U$ is a symmetry if it stabilizes the state whose coefficients are given by the tensor $\mathbf{V}$. From now on we use $\mathbf{V}$ to denote tensors with components $V_{i_1,\dots, i_n}$ where we will assume that a canonical choice of basis $\mathbf{e}_{i_1,\dots,i_n}$, such as the computational basis, is used. In this work, we do not distinguish upper and lower indices (or basis and dual basis).
For example, if the atomic Lego is a Pauli stabilizer code, then the relevant symmetries contain its stabilizer and normalizer elements. This is the case, for instance, if $U$ is a Pauli string or $\mathcal{O}_i$'s are Pauli operators in Fig.~\ref{fig:qlego}b. In that case, these symmetries also uniquely determine the tensor $\mathbf{V}$, which allows us to design codes without ever discussing the tensor components $V_{i_1,\dots, i_n}$.


\begin{figure*}
    \centering
    \includegraphics[width=0.76\linewidth]{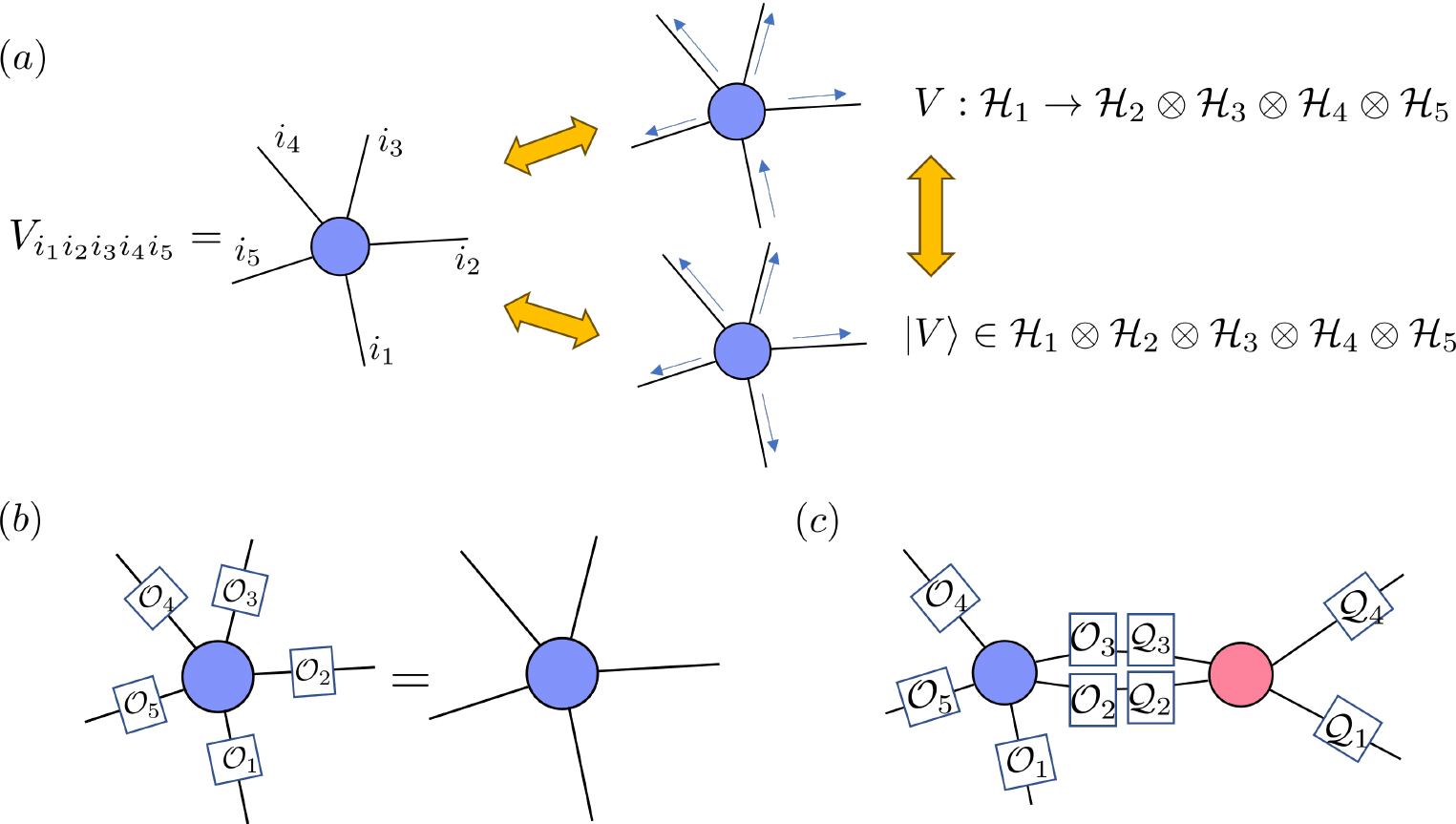}
    \caption{(a) A tensor $V_{i_1i_2i_3i_4i_5}$ can be expressed graphically where each index is assigned to a dangling leg. The same tensor can be used to construct states or maps but these objects are related to each other via the channel state duality (yellow arrow). The in-going blue arrow marks the input/logical degree of freedom while the out-going blue arrows mark output/physical degrees of freedom. (b) A tensor has a symmetry if it is invariant after contracting with operators. (c) If the local symmetries on the blue and red tensors, marked by $\mathcal{O}$'s and $\mathcal{Q}$'s respectively, are matching, then the operators acting on the connected edges satisfy $\mathcal{Q}_3=\mathcal{O}_3^*$ and $\mathcal{Q}_2=\mathcal{O}_2^*$. It is customary to drop the boxes when writing the symmetries to avoid clutter.}
    \label{fig:qlego}
\end{figure*}

These atomic blocks can then be joined together by graphically connecting their legs. This corresponds to tensor contractions in the tensor network, which is colloquially known as ``tracing''\footnote{We use $\wedge_{j,k}$ to denote tensor contraction of indices $j,k$. It should be clear from context that it is not the wedge product. }. 
\begin{definition}
    Consider a tensor $\mathbf{V}$ with components $V_{i_1,\dots i_n}$. We denote tracing (or tensor contraction) that connects wires $i_j,i_k$ as an operation $\wedge_{j,k}$ such that
    $$\wedge_{j,k}\mathbf{V} \leftrightarrow \sum_{i_j,i_k}V_{i_1,\dots, i_n}\delta_{i_j,i_k},$$
    where $\delta_{ab}$ is the Kronecker delta.
\end{definition}

Graphically, tracing simply joins the wires corresponding to the indices that are summed over. Note that $\mathbf{V}$ could be a tensor product of two or more uncontracted tensors, in which case tracing can join together disconnected components.

Physically, if each atomic Lego corresponds to a state, then (up to normalizations) connected edges correspond to Bell fusions, where the connected physical qubits are projected onto a Bell state, and then discarded. This is summarized as the following lemma.
\begin{lemma}\label{lemma:trace}
Let $|V\rangle, |W\rangle$ be the state defined by tensor $\mathbf{V}, \mathbf{W}$ respectively, then tracing corresponds to
    \begin{equation}
    \mathbf{W}=\wedge_{j,k} \mathbf{V} \leftrightarrow |W\rangle=\langle \Phi^+|V\rangle = \wedge_{j,k}|V\rangle
    \end{equation}
    where $|\Phi^+\rangle = |00\rangle+|11\rangle$ is supported on tensor factors $j, k$.
\end{lemma}

The symmetries of these joined Lego blocks are then generated through operator pushing (see Sec 2.3 of~\cite{Cao2022}) or operator matching. 
\begin{theorem}[Operator Matching]\label{thm:opmatch}
  {\color{black}  Let $U\otimes O_{j}\otimes O_{k}\otimes U'$ be a unitary symmetry of $\mathbf{V}$ where $O_j, O_k$ act on the $j$th and $k$th leg/qubit respectively. Suppose $O_{j}=O_{k}^*$ where $*$ denotes complex conjugation, then $U\otimes U'$ supported on the Hilbert space except factors $j$ and $k$ is a symmetry of $\wedge_{j,k}\mathbf{V}$.}
\end{theorem}
See ~\cite{Cao2022,Cao2022_2} for more detailed discussions and proofs for Lemma~\ref{lemma:trace} and Theorem~\ref{thm:opmatch}. 
An example of this is shown in Fig.~\ref{fig:qlego}c where $\mathcal{O}_i$'s and $\mathcal{Q}_i$'s denote the symmetry of the blue and red tensors respectively. If $\mathcal{Q}_3=\mathcal{O}_3^*$ and $\mathcal{Q}_2=\mathcal{O}_2^*$, then the remaining operators acting on the dangling edges is a symmetry of the contracted tensor network. By repeating this matching process for all the local symmetries of each tensor, we obtain a set of symmetries for the traced tensor network. We use $O^*$ to denote the complex conjugate of $O$ and $O^{\dagger}$ for the conjugate transpose in this work.

In PSF, the symmetries generated by matching the stabilizers of each local tensor are sufficient to uniquely identify any code constructed by fusing the atomic blocks. These matching symmetries can also be tracked and generated efficiently over check matrices with an operation called conjoining. Suppose $H$ is a check matrix of the stabilizer state, or that of the dual stabilizer Choi state defined by a Pauli stabilizer encoding map, then for prime $q$, the conjoining operation is a map $\wedge_{j,k}: \mathbb{F}_q^{n\times 2n}\rightarrow \mathbb{F}_q^{(n-2)\times 2(n-2)}$ from the space of check matrices over $n$ qudits to the space of check matrices over $n-2$ qudits\footnote{For clarity, here we have dropped the phase column in the stabilizer check matrix.}.  See App. D of~\cite{Cao2022}. Because this operation can be performed with row and column operations over the check matrices, they are classically tractable even for large $n$.

Finally, recall that the legs themselves may be interpreted as either physical qubits (output legs) or logical qubits (input legs) (Fig.~\ref{fig:qlego}a right). {\color{black}If the unitary symmetry of the tensor admits a product structure, then the same symmetry may represent a stabilizer or a logical operator depending on how the legs are assigned.

Let us partition the indices of $V_{i_1,\dots, i_n}$ into two complementary sets $J$ and $J^c$, to which we can assign the Hilbert spaces $\mathcal{H}_{J}$ and $ \mathcal H_{J^c}$ respectively. We discussed above that a state $|V\rangle$ can be obtained from the tensor, but so can a linear map $V:\mathcal{H}_{J^c}\rightarrow \mathcal{H}_J$ such that

\begin{equation}
    V=\sum_{j_a\in J,k_b\in J^c} V_{j_1,j_2,\dots, k_1,k_2,\dots} |j_1,j_2\dots\rangle\langle k_1, k_2\dots|.
\end{equation}

The Schmidt coefficients of $|V\rangle$ can now be related to the (nonzero) singular value spectrum of $V$ up to normalization factors. For this work, we will focus on the maps that are isometric, thus representing encoding maps. Without loss of generality, let $|J^c|\leq |J|$. If $|V\rangle$ is maximally entangled between $J$ and $J^c$, then $V$ is an isometry.

 \begin{proposition}\label{prop:3.1}
     Suppose $V$ is an isometry and $U=O_J\otimes Q_{J^c}$ is a unitary symmetry of its Choi state $|V\rangle$, then the code admits a logical operator such that $\bar{Q}^t=O_J$.
 \end{proposition}
This implies that given a $|V\rangle$ that is sufficiently entangled across certain bipartitions, one can obtain different encoding maps even from the same TN in a way similar to code shortening. 

For example, consider a $[[4,2,2]]$ code whose stabilizer group is generated by $XXXX$ and $ZZZZ$ and has logical operators $\bar{X}_1=XXII, \bar{Z}_1=ZIZI, \bar{X}_2=XIXI$, and $\bar{Z}_2=ZZII$. The tensor representation has six legs, two of which correspond to the two logical qubits while the other four the physical qubits. If we re-designate the first logical leg as a physical qubit instead, then it becomes a $[[5,1,2]]$ code with stabilizer group $S=\langle XXXXI, ZZZZI, XXIIX, ZIZIZ\rangle$ and logical operators $\bar{X}=XIXII, \bar{Z}=ZZIII$. On the other hand, if we re-assign one of the physical indices in the $[[4,2,2]]$ code as logical degrees of freedom, then the same tensor can be used as an encoding isometry for a trivial $[[3,3,1]]$ code\footnote{For more explicit examples involving the $[[4,2,2]]$ CSS code, see Appendix A of \cite{Su2023}.}. 

For more complex example that moves beyond code shortening, Ref.~\cite{Cao2022} shows that the surface code and the 2d Bacon-Shor code are interchangeable via this duality where $V$ also need not be an isometry. However, in this work, we only focus on the cases where the encoding map $V$ is an isometry. More often, we only focus our discussion on states that satisfy the above entanglement criteria, because they can be turned into a map in such a way.}

\subsection{Fusing XP Legos and operator matching}
\label{subsec:xpfusion}
Now we examine how the desirable traits of QL extend to XP codes. When tracing a state, \cite{Cao2022} showed that the symmetries generated from operator matching remain a symmetry of the joint state. While this symmetry from matching is sufficient when tracing Pauli stabilizer states, the same need not be true when tracing XP states. Here we provide the conditions for when they are enough. 

First we note a key difference from the XPF in that tracing XP legos no longer needs to remain XP. This can be shown by counterexample in App.~\ref{app:tracingXP} and in Ref.~\cite{Webster2022}, where we notice that XP measurements are equivalent to tracing single-qubit XP states. Intuitively, this is because for $N>2$, \textcolor{black}{XP states include magics states} in addition to Pauli stabilizer states. For example, the magic state $\ket{H}=\ket{0}+e^{i\pi/4}\ket{1}$ is an XP state, stabilized by $XP_8(1|3|2)$. Together with the encoding tensors of the repetition code (GHZ state), the Hadamard tensors and the rank 1 tensors from $|0\rangle$, they can be fused to form a universal gate set and measurement with post-selections. 
As a result, their fusions can produce states dense in the Hilbert space (Fig.~\ref{fig:XPprocess}) as shown by Theorem 2.3 of~\cite{Cao2022}. 
\begin{figure}
    \centering
    \includegraphics[width=0.4\textwidth]{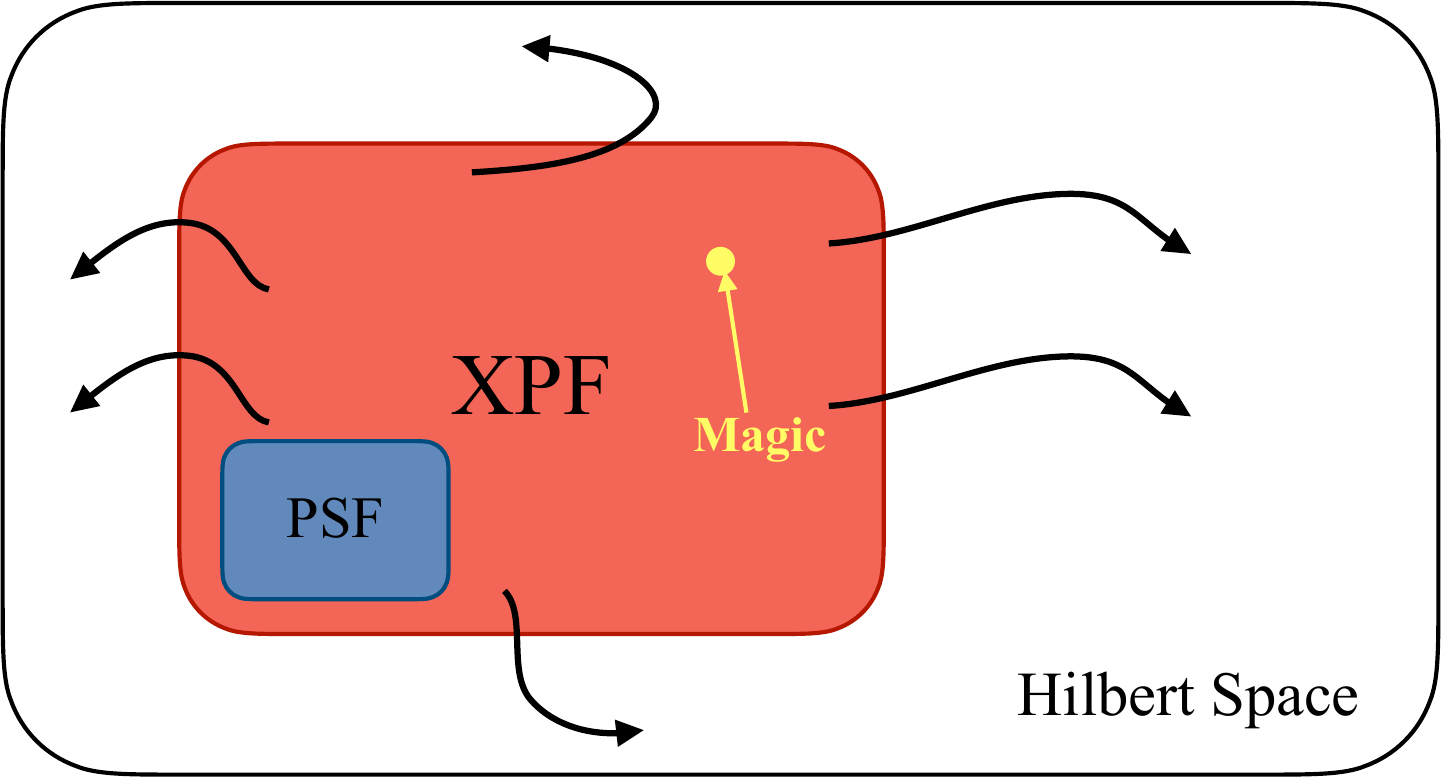}
    \caption{Pauli stabilizer states (blue box) only cover a fraction of the whole many-body Hilbert space. In contrast, XP states with arbitrary precision (red box) encompass not only Pauli stabilizer states but also the magic state. When combined with the Clifford operations found within PSF, they constitute a universal gate set. Therefore, an XP Lego set is universal.}
    \label{fig:XPprocess}
\end{figure}
Since the XP states by themselves are not dense in the Hilbert space (Lemma~\ref{lemma:notdense}), even though they are far more numerous than Pauli stabilizer states, there must exist fusions of XP states that are not XP. 

\begin{theorem}
    Tracing an XP stabilizer state with precision of $N>2$ does not always produce an XP stabilizer state for any precision.
\end{theorem}

\begin{corollary}
    Tracing XP stabilizer states at a fixed precision $N=2^t>2$ does not always produce an XP stabilizer state of the same precision.
\end{corollary}

Conversely, it is also possible to produce an XP state by tracing non-XP states --- consider an XP state that is produced by tracing two XP states. We now rotate one wire of the Lego by $U$ which takes it out of the XP state space. We then rotate the other Lego by $U^*$.

Therefore, we do not expect operator matching to be able to uniquely identify the state in general. From the complexity theoretic perspective, these statements are unsurprising because, as we will see later, operator matching (or check matrix conjoining) can be performed classically with an efficient algorithm. Since tracing enables universal quantum computation, such operator matching should not be able to specify all outcomes of a quantum computation efficiently.

For the present work, we would like to focus on the cases where trace does return an XP state of the same precision instead of a general quantum state. The reason for this restriction is two-fold: 1) symmetries present in an XP state (but lacking in a general state) allow us to identify fault-tolerant gates. 2) The symmetries of the XP states permit a classically tractable check matrix description over binary fields and finite rings, similar to~\cite{GottesmanKnill,AaronsonGottesman}. Note that Theorem~\ref{counting theorem} provides a necessary condition for when the state is XP. 

Like the Pauli case, we can track the symmetries, or more precisely, the LID group of the post-trace state using the conjoining operation over XP check matrices. We give an algorithm for this procedure below: 


\begin{algorithm}[t]
    \begin{algorithmic}
        \STATE{$\mathbf{M}_Z=[],LID=[]$}
        \STATE CanonicalForm($\mathbf{S}$)
        \FOR{$r,r'\in S$}
        \IF{the restriction of $r, r'$ to the first two qubits has non-trivial $X$ support on only the 1st or the 2nd qubit respectively}
        \STATE $r\leftarrow rr'$, remove $r'$ from $\mathbf{S}$
        \ENDIF
        \ENDFOR
        \FOR{$r\in \mathbf{S}_Z$}
        \IF{$r$ acts on the first two qubits}
        \STATE{$\mathbf{M}_Z \gets Append(\mathbf{M}_Z ,r)$}
        \ENDIF
        \ENDFOR

        \FOR{$r \in \mathbf{S}$}
        \IF{$Match(r,\mathbf{M}_Z )$ exists}
        \STATE{$LID \gets Append(LID,Match(r,\mathbf{M}_Z))$}
        \ENDIF
        \ENDFOR

        \STATE CanonicalForm($LID$)
    \end{algorithmic}
        \caption{Operator matching on check matrix}
    \label{algo:1}
\end{algorithm}
\textcolor{black}{This algorithm updates the LID check matrix of a state after a self-trace on two qubits by constructing new generators that satisfy a specific matching condition. Without loss of generality, we take the self-trace to be on the first two qubits. The procedure begins with a pre-processing step on the initial check matrix, $\mathbf{S}$, to consolidate X-supports. It iterates through pairs of generators, and if one has an $X$ on the first qubit  and the other has an $X$ on the second , they are multiplied into a single generator. This ensures that any X-support on the traced qubits appears as $X_1X_2$ or $I_1I_2$ within a single row. Following this, the algorithm enforces the full matching condition for an operator $X^{a_1}P^{b_1} \otimes X^{a_2}P^{b_2}$ acting on the traced qubits, which requires $a_1=a_2$ and $b_1+b_2=0 \pmod N$. To satisfy the condition on the $P$-parts, we identify the diagonal generators $\mathbf{M}_Z$ acting on the traced qubits. For each generator $r$ in the pre-processed set, we find a combination from $\langle \mathbf{M}_Z \rangle$ that modifies $r$ to meet the matching criterion. Finally, all resulting non-trivial matched generators form the new LID after removing the columns for the traced qubits and converting the matrix to its canonical form.}

Like in the Pauli case, this algorithm is polynomial in $n$ for any fixed $N$. Repeated applications of Algorithm \ref{algo:1} as a subroutine allow us to produce the XP symmetries of the resulting tensor network. Note that the dominant cost of the subroutine is limited by the conversion into the canonical form~\cite{howell_ring}.

Consider the set of quantum gates or processes dual to XP states, which we call \emph{XP processes}. For general $N=2^t$, this also includes any Clifford gate and stabilizer measurements. However, as we noted earlier, the inclusion of higher precision XP states also includes non-trivial phase gates and elements from higher levels of the Clifford hierarchy. For the subset of these XP processes that map XP states to XP states, conjoining also provides a classically efficient method for simulating these dynamical processes by representing them as operations over finite fields and rings. This is similar to Ref.~\cite{GottesmanKnill}, but also accommodates non-(Pauli)-stabilizer states. Thus it can be considered a partial extension of the traditional stabilizer simulation. It is conceivable that a full extension of such methods can lead to a more comprehensive simulation algorithm beyond PSF and those based on Abelian groups~\cite{AaronsonGottesman,Bermejovega2013}. We leave this to future work.

Finally, none of the above would have been effective if conjoining or operator matching did not fix the post-trace XP state. 
\begin{theorem}\label{Operator matching is enough}
Suppose $|W\rangle=\wedge_{j,k}|V\rangle$ such that both $|W\rangle, |V\rangle$ are XP states, then operator matching on the logical identity group of $|V\rangle$ produce the entire logical identity group of $|W\rangle$.
\end{theorem}

This result is key for QL as it implies that the language of operator pushing continues to serve as a useful graphically intuitive and computationally efficient guide for building non-Pauli stabilizer codes. 

An example of the conjoining operation is shown below {\color{black} where we glue together two physical legs of the tensor to make a $5$-qubit code}.
Consider the stabilizer group generators for a $[[7,2,2]]$ code defined with precision $N=8$:

\begin{equation} \label{eqn:check7qbit}
\scalebox{0.8}{$
\left( \begin{array}{ccccccc|ccccccc|c}
1 & 1 & 1 & 0 & 0 & 0 & 0 & 0 & 0 & 7 & 0 & 0 & 0 & 0 &  9\\
0 & 0 & 0 & 1 & 1 & 1 & 1 & 0 & 0 & 0 & 1 & 2 & 3 & 4 & 14\\
0 & 0 & 0 & 0 & 0 & 0 & 0 & 1 & 0 & 7 & 0 & 0 & 0 & 0 &  0\\
0 & 0 & 0 & 0 & 0 & 0 & 0 & 0 & 1 & 7 & 0 & 0 & 0 & 0 &  0\\
0 & 0 & 0 & 0 & 0 & 0 & 0 & 0 & 0 & 0 & 4 & 4 & 4 & 4 &  8
\end{array}\right).$}
\end{equation}

\textcolor{black}{We will perform a self-trace on the first two qubits, and derive the XP stabilizer for the post-tracing $5$-qubit code by following Algorithm~\ref{algo:1}. A self-trace here is a check matrix operation that one induces by connecting the two legs (or contracting two indices) on the same tensor. It is defined in detail for Pauli stabilizer codes in App. D of \cite{Cao2022}.}

\textcolor{black}{First, since the $X$ part is already matched, we search for diagonal operators that act non-trivially on the first two qubits. As stated in the text, rows 3 and 4 of the original check matrix satisfy this condition, so there are two elements in $\mathbf{M}_Z$.}

\textcolor{black}{Then, we try to match each row within $\mathbf{M}_Z$. Recall that the matching condition of the check matrix comes from operator matching in the tensor network. Since we work with XP operators $X^{a_1}P^{b_1}$ and $X^{a_2}P^{b_2}$, two such operators match if $a_1=a_2$ and $b_1+b_2=0~\mathrm{mod}~N$.
Therefore, the first and second rows satisfy the matching condition.
For the third and fourth rows, they need to be matched with each other to cancel the $P$ operator on the traced leg. One can do so by multiplying $-1$ on the fourth row and add, but that will lead to a trivial vector, i.e., an identity operator. The fifth row satisfies the matching condition trivially.}

\textcolor{black}{The collection of non-trivial operators resulting from this matching process thus forms the new LID group. Since self-trace will annihilate two physical qubits by projecting them onto the Bell state and discarding them, we remove the 1st, 2nd, 8th, and 9th columns of the check matrix (\ref{eqn:check7qbit}) that correspond to these two qubits. Finally, we obtain the check matrix for the post-trace code:} 
\begin{equation}
\left( \begin{array}{ccccc|ccccc|c}
 1&  0&  0&  0&  0&  7&  0&  0&  0&  0&  9\\
 0&  1&  1&  1&  1&  0&  1&  2&  3&  4& 14\\
 0&  0&  0&  0&  0&  0&  4&  4&  4&  4&  8
\end{array}\right).
\end{equation}

\subsection{Atomic XP Legos}

Recall that tracing the atomic Legos $S,H,|0\rangle$ and the 2 qubit repetition codes cover the entire set of Pauli stabilizer codes. We may ask whether a similar set can be identified for the XP states (or XP regular codes). While we note that the addition of any tensor derived from any non-Clifford gate would be sufficient to approximate any state to an arbitrary degree of accuracy, it still makes sense to ask whether we can identify the minimal element needed to reproduce the XP states exactly. 

\begin{theorem}
    Any XP state of precision $N$ can be produced by atomic Legos $P,H,|0\rangle$ and the 2 qubit repetition codes where $P=diag(1, \exp(i\pi/N))$. 
\end{theorem}

\begin{proof}
    We use the fact that any XP state of precision $N$ can be mapped to weighted hypergraph states with generalized control phase operators $\{CP(p/2N,\mathbf{v})\}$, where the phase factor is $\exp(2i\pi p/2N)$ when acting on $|\mathbf{e}\rangle$ and $\mathbf{ev}=\mathbf{v}$ and $1$ otherwise. Note that with $|0\rangle, H$ and repetition codes, one can prepare $|+\rangle$ and CZ gates. It remains to show that we can prepare such generalized control phase gates needed for the weighted hypergraph states. 

    This is possible with a prescription by~\cite{SchuchSiewert} where one alternates control-nots (which can be built from CZ and H) and phase gates $P=diag(\exp(\mp i\phi/2), \exp(\pm i\phi/2))$, (or $P=diag(1, \exp(\pm i\phi))$ if shifted by a global phase that does not affect the outcome of the computation. The resulting generalized CP can be chosen to have a non-trivial phase factor that is any integer multiple of $\phi/2$ by inserting phase gates at the right places. In our case, it is sufficient to have $\phi= \exp(i\pi/N)$.
\end{proof}

Note that for precision $N=2$, $P=Z$ these are exactly the Pauli stabilizer atomic Legos. However, it is clear from the previous theorems that tracing these components can take us out of the set of XP codes with precision $N$ as long as $N>2$. 

Given the set of atomic Legos, another pertinent question is how many such components are needed to prepare a tensor network that describes a state or code of interest. While an efficient method is still under exploration, one can provide rough estimates based on gate synthesis through tensor tracing. As mentioned above and in~\cite{Cao2022} one can construct unitary gates by contracting tensors. This seems to indicate that in order to prepare states with exponential circuit complexity, exponentially many atomic tensors may be required. This, however, need not be true, as the tensor network is a classical description that includes post-selections. Therefore, a state with exponential circuit complexity may be prepared with polynomial time with post-selection. Such a measurement-based state preparation process can then be converted into a tensor network. Although this TN description of the state can be exponentially hard to contract, the number of atomic tensors remains polynomial in the number of qubits.

\section{Optimal Decoders}\label{Section4}
Error correction in the Pauli stabilizer codes proceeds by measuring the commuting check operators or stabilizer generators. Depending on the error patterns, one obtains a list of syndromes $s$. Based on the syndrome information, a decoder is then run to determine the correction operator to be applied to the code. The error is corrected if the overall error and correction process implements a logical identity in the code. However, because the XP stabilizer group is non-Abelian, both the syndrome extraction and the traditional stabilizer decoding process may need to be reworked.

Fortunately, this challenge is not unique to the XP stabilizer codes. For example, it is known that multiple rounds of decoding can be implemented in topological codes with non-Abelian symmetries~\cite{nat} by decomposing and decoding a sequence of Abelian stabilizer codes. In our case, notice that the dihedral group admits a cyclic normal subgroup generated by the diagonal generator $P$ and hence a short exact sequence $1\rightarrow \mathbb{Z}_N\rightarrow D_N\rightarrow \mathbb{Z}_2\rightarrow 1$ in the language of~\cite{nat}. Independently, a similar two-step routine is suggested by~\cite{Webster2022} and~\cite{Ni2015}. In these decoding schemes, one first decodes the code $\mathcal{C}_Z$ stabilized by the Abelian diagonal subgroup of $\mathcal{S}$, then performs any recovery operation needed to restore the state back to $\mathcal{C}_Z$. In the second round, the non-diagonal checks, which commute within $\mathcal{C}_Z$ are measured and a decoding scheme similar to the Abelian code follows. Although the outline of the syndrome extraction process is known, no explicit decoder has been proposed to the best of our knowledge. For the error correction process, we first follow the syndrome extraction outlined by~\cite{Webster2022}, we then construct a maximum likelihood decoder for XP-regular codes under any i.i.d. single qubit error channels by building upon the recent work by~\cite{Cao2023}. We show that the general enumerator method can be fully extended to decode XP codes in a way similar to Pauli stabilizer codes. As such, the efficiency is comparable and would be largely limited by the hardness of (approximate) tensor network contractions. 

Assuming that we are given the QL construction of the XP code, i.e., we have a tensor network definition of its encoding map along with its check matrix in the canonical form by performing the conjoining operations such that generators $\mathbf{S}_X, \mathbf{S}_Z$ are known. From $\mathbf{S}_Z$, \cite{Webster2022} shows that one can identify a set of diagonal Pauli generators $\mathbf{R}_Z$ such that $$\Pi_Z=\frac{1}{|\langle\mathbf{R}_Z\rangle|}\sum_{s\in\langle \mathbf{R}_Z\rangle}s = \frac{1}{|\langle\mathbf{S}_Z\rangle|} \sum_{s\in\langle \mathbf{S}_Z\rangle}s,$$i.e., the two Abelian stabilizer codes stabilized by the two diagonal groups coincide.

Let us consider the physical error model that acts identically across all qubits such that the error channel on each qubit $\rho_j$ is given by the Kraus operators:

\begin{equation}
    \rho_j\rightarrow \sum_{i=1}^4 K_i \rho_j K_i^{\dagger}.
\end{equation}
Let us write the error channel $\mathcal{E}$  that acts on the entire system $\rho$ as 
\begin{equation}
   \mathcal{E}(\rho)= \sum_{\mathbf{i}} \mathcal{K}_{\mathbf{i}} \rho \mathcal{K}_{\mathbf{i}}^{\dagger}
\end{equation}
where 
$$\mathcal{K}_{\mathbf{i}}=K_{i_1}\otimes K_{i_2}\otimes \dots\otimes K_{i_n},$$ and $\mathbf{i}$ is over all quaternary strings of length $n$.

For the sake of conceptual clarity, let us perform the following two-step syndrome extraction (as shown in Fig.~\ref{fig:decoder}) 
\begin{figure}
    \centering
    \includegraphics[width=0.46\textwidth]{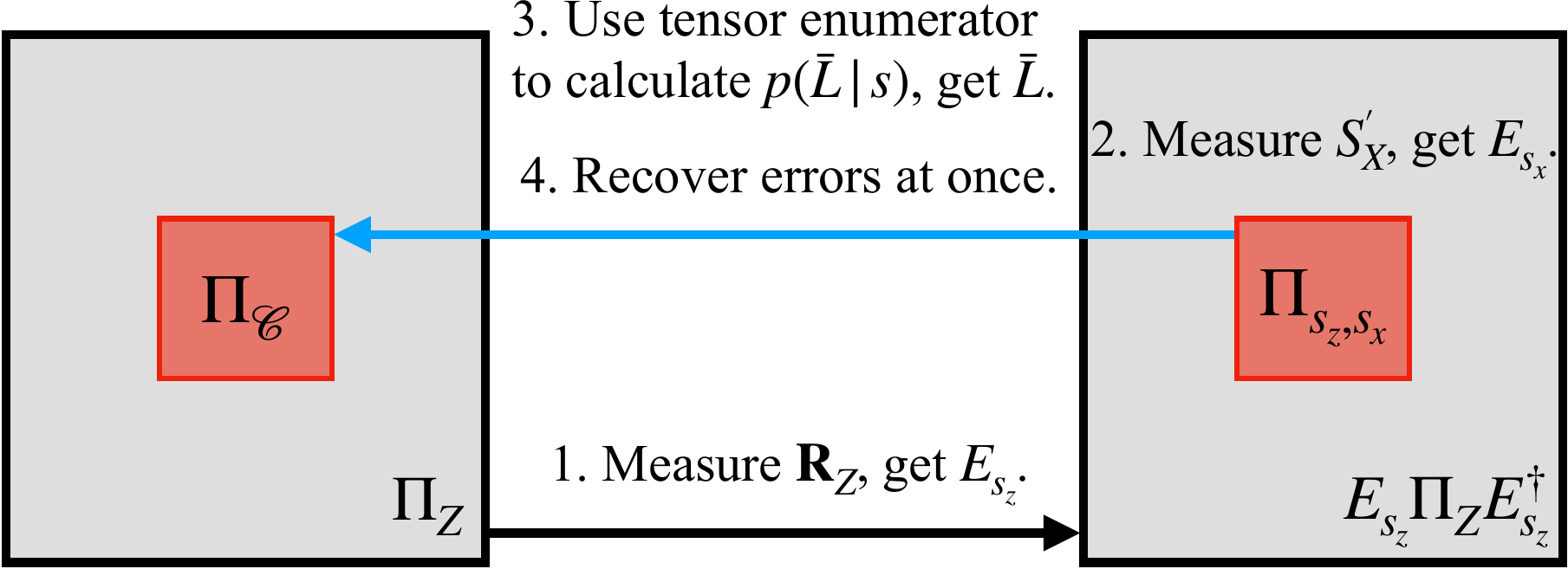}
    \caption{Illustration of the two-round decoder.}
    \label{fig:decoder}
\end{figure}
where for now we overlook its practical implementation, which we explain near the end. We start by measuring the Pauli checks $\mathbf{R}_Z$, which define a Pauli stabilizer code. From the syndromes $s_z$, one can identify any $E_{s_z}$ in the coset. We then correct the error $E_{s_z}$ and reset the code into the subspace supported on $\Pi_Z$. Next, we measure the generators $\mathbf{S}_X$, obtaining its syndromes $s_x$. It is easy to show that one can always determine a diagonal error operator $E_{s_x}$ with such a syndrome. This measurement projects the state onto the subspace supported by $E_{s_x}\Pi_{\mathcal{C}} E^{\dagger}_{s_x}$, where $\Pi_{\mathcal{C}}$ is the projection onto the XP code subspace defined in the code projector (\ref{eqn:codeproj}). 
We then restore the state to the code subspace by applying $E_{s_x}^{\dagger}$.

For $|\tilde{\psi}\rangle\in \mathcal{C}$, the probability that the above process incurs a logical error $\bar{L}$ is then
\begin{align}
    &p(\bar{L}\cap s_{x,z},E_{s_{x,z}}) = \\\nonumber
    &\int d\mu(|\tilde{\psi}\rangle)\sum_{\mathbf{i}}\Tr[\bar{L}|\tilde{\psi}\rangle\langle\tilde{\psi}|\bar{L}^{\dagger}E_{s_x}^{\dagger}\Pi_{s_x}E_{s_z}^{\dagger}\Pi_{s_z}\mathcal{K}_{\mathbf{i}} |\tilde{\psi}\rangle\langle\tilde{\psi}|\\\nonumber
    &\times\mathcal{K}_{\mathbf{i}}^{\dagger}\Pi_{s_z}E_{s_z}\Pi_{s_x}E_{s_x}],
\end{align}
where the projection onto non-trivial syndrome subspace is related to that onto the trivial syndrome subspace as $\Pi_{s_z}=E_{s_z}\Pi_Z E^{\dagger}_{s_z}$ and $\Pi_{s_x}=E_{s_x}\Pi_{\mathcal{C}} E^{\dagger}_{s_x}$. The integral is evaluated over the normalized uniform measure. Plugging in these expressions and simplify by noting that $E_{s_x}$ is diagonal and hence commutes with $\Pi_Z$,  we obtain

\begin{align}
    &p(\bar{L}\cap s_{x,z},E_{s_{x,z}}) 
    =\frac{1}{K(K+1)}\Big(\sum_{\mathbf{i}}\Tr[\mathcal{K}_{\mathbf{i}} \Pi_{\mathcal{C}} \mathcal{K}_{\mathbf{i}}^{\dagger} \Pi_s] \\\nonumber
    &+\sum_{\mathbf{i}}\Tr[\mathcal{K}_{\mathbf{i}}\Pi_{\mathcal{C}}\tilde{E}_s^{\dagger} ] \Tr[\mathcal{K}_{\mathbf{i}}^{\dagger}\tilde{E}_s\Pi_{\mathcal{C}}]\Big)
\end{align}
where $\tilde{E}_s=E_{s_z}E_{s_x}\bar{L}$, and $\Pi_s= \tilde{E}_s\Pi_{\mathcal{C}} \tilde{E}_s^{\dagger}$. We recognize the above expression is nothing but
\begin{align}
    &p(\bar{L}\cap s_{x,z},E_{s_{x,z}})\\\nonumber
    = &\frac{ B(\mathbf{k}; \Pi_{\mathcal{C}}, \Pi_s) + A(\mathbf{k}; \Pi_{\mathcal{C}} \tilde{E}_s^{\dagger},  \tilde{E}_s\Pi_{\mathcal{C}}) }{K(K+1)},
\end{align}
where $A(\mathbf{u};M_1,M_2),B(\mathbf{u};M_1,M_2)$ are the generalized coset enumerators introduced by~\cite{Cao2023} and $\mathbf{k}=\{k_{P{P'}}\}$ are the coefficients by expanding the Kraus operators of the single qubit error channel in the Pauli basis,
$$\sum_i K_i\cdot K_i^{\dagger}=\sum_{P{P'}} k_{P{P'}} P\cdot {P'}.$$
For detailed definitions, see Sections 2 and 3 of \cite{Cao2023}.

As long as $\tilde{E}_s$ is a tensor product of single qubit operators, any enumerators of the above form can be evaluated using the tensor enumerator technique given a QL construction of the code. Since $p(\bar{L}|s_{x,z},E_{s_{x,z}}) = p(\bar{L}\cap s_{x,z},E_{s_{x,z}})/p(s_{x,z},E_{s_{x,z}})$ and $p(s_{x,z},E_{s_{x,z}})$ is a constant for a fixed syndrome, we then complete the maximum likelihood error correction process by undoing $\bar{L}E_{s_x}$  for which $p(\bar{L}\cap s_{x,z},E_{s_{x,z}})$ is maximized. 

In practice, it is equivalent (in the Heisenberg picture) to measure the checks $\mathbf{R}_Z$, immediately followed by $\mathbf{S}_X'=E_{s_z}\mathbf{S}_X E_{s_z}^{\dagger}$ instead of performing an active recovery operator $E_{s_z}$ in between. Then the error can be undone all at once at the very end by applying $\tilde{E}^{\dagger}_s$. Note that although $\mathbf{S}_X, \mathbf{S}'_X$ are generally non-Hermitian, they only have eigenvalues $\pm 1$ in the subspace supported on $\Pi_Z$, acting like Pauli operators. Hence one can use the conventional syndrome extraction circuit for the second round of syndrome measurements. Further details of the decoder are found in App.~\ref{app:xpdecoder}.

The efficiency of the decoder depends on the architecture of the tensor network that builds up the XP code. Since the exact optimal decoding for stabilizer codes is $\#$P complete, the complexity of ML decoding generic codes is necessarily exponential in the system size. This is consistent with the observation that the exact contraction of arbitrary tensor networks is hard. Therefore, the method can provide an efficient decoder for certain systems where the network contraction is efficient. For others, it still provides subexponential or polynomial speed up compared to brute force evaluation. See Sec. 4 of~\cite{Cao2023}. For example, the tensor network construction of the twisted quantum double is known~\cite{Guetal, Vidal}. Therefore, the above construction yields the optimal decoder for the double semion model where the exact ML decoder has complexity $O(\exp(\sqrt{n}))$, which is subexponential in the system size for any i.i.d. single site error. In practice, the exact contraction of the tensor network needed for such a theoretical ML decoder is seldom necessary; by permitting approximate contractions when computing the conditional error probabilities, additional speedup may produce efficient decoders similar to~\cite{MLTNS,Chubb}. {\color{black} For more general XP codes that are not built from quantum lego blocks, as long as the check operators are given, then the code can also be decomposed into smaller tensors using the construction in Sec. IV of \cite{Cao2023} which works for non-Abelian codes. Constructing and contracting their enumerator tensor network will then provide the optimal decoder for these XP codes also.}

\section{Novel XP codes}\label{Section5}
Although twisted quantum double models like the double semion model are known to produce non-trivial XP codes, there are very few known examples of XP codes with $N>4$ that are not also known to be Pauli stabilizer codes, e.g. triorthogonal codes. Furthermore, it is difficult to verify whether a given XP code is not simply related to a Pauli stabilizer code by local unitary (LU) deformations. 
For example, the only explicit example of a non-trivial XP regular code that is not known to violate the above criteria is a $[[7,2,2]]$ error detection code discovered by~\cite{Webster2022}.

As a proof of principle, here we repeat the simple moves in~\cite{Cao2022} to construct another example with a higher distance. We first identify a smaller atomic $[[4,2,2]]$ XP code and then a $[[7,1,3]]$ XP code by contracting the $[[4,2,2]]$ Legos. 
Recall that each regular XP code has a correspondence with a CSS code. To construct the 4-qubit XP code, we start from the codewords of the conventional $[[4,2,2]]$ CSS code by introducing additional phases. 
Through exhaustive search\footnote{We search through $16^7$ phase sets, ultimately identifying a subset of $4194304$ phase sets that are consistent with a $[[6,0,3]]$ XP state.}, we craft the codewords for various XP $[[4,2,2]]$ tensors of precision $N=8$. Alternatively, one can deform the 4-qubit CSS codes with local unitaries and use them as atomic Legos.

To create a larger code, we connect two such atomic Legos. Note that although these $[[4,2,2]]$ XP codes may be LU equivalent to the $[[4,2,2]]$ CSS code, as their enumerators indeed coincide, the contraction of LU deformed Legos need not be LU equivalent to the contraction of undeformed legos. A simple example is the $[[5,1,3]]$ code obtained from contracting two LU-deformed $[[4,2,2]]$ codes, whereas the undeformed contraction can never attain the perfect code~\cite{Cao2022_2,Su2023}. A simple reason behind this observation follows from operator pushing --- unless the LUs inserted in the contracted legs can be represented as transversal gates on the external legs, they are generally pushed to quantum gates that act non-locally on the physical qubits, thereby altering the code properties relative to the undeformed contractions. As a result, we expect this strategy to be able to generate non-trivial XP codes despite applying only local deformations on the CSS Legos.

Taking a cue from a similar process demonstrated in Ref.~\cite{Cao2022}, where two $[[4,2,2]]$ tensors were joined to create the $[[7,1,3]]$ Steane code, we follow a similar tensor network contraction by duplicating the $[[4,2,2]]$ tensors and then linking their ``logical legs'' \textcolor{black}{(the legs pointing up and down)}. An encoding map for the 7 qubit code is defined by designating a leg as the logical qubit while the rest as physical qubits (Fig.~\ref{fig:steane}).
\begin{figure}
    \centering
    \includegraphics[width=0.7\linewidth]{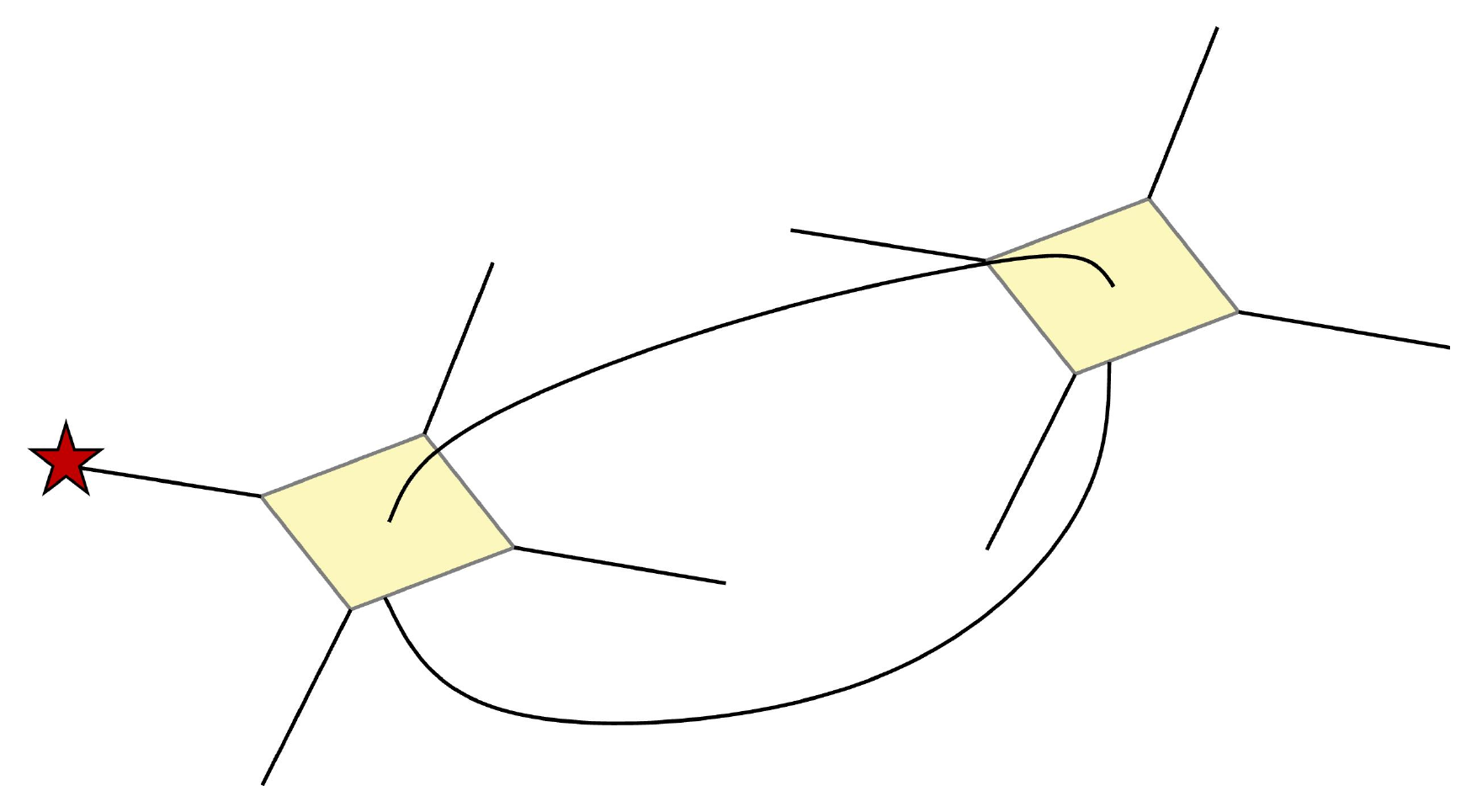}
    \caption{Two $[[4,2,2]]$ XP codes are contracted along the same legs to produce a rank-8 tensor. The logical leg is marked by a star. }
    \label{fig:steane} 
\end{figure}

Here we give two types of $[[7,1,3]]$ XP code with precision $N=8$ that can not be represented in PSF.
The first example is similar to the Steane code.
The check matrix for the 6-leg state from the $[[4,2,2]]$ tensor is 
\begin{equation}\left(\begin{array}{cccccc|cccccc|c}\label{eq:checkmatrixV1}
1& 0& 0& 1& 1& 1& 0& 0& 0& 0& 0& 0& 0\\
0& 1& 0& 1& 0& 1& 0& 0& 0& 0& 0& 0& 0\\
0& 0& 1& 1& 1& 0& 0& 0& 3& 4& 0& 4& 1\\
0& 0& 0& 0& 0& 0& 4& 0& 0& 4& 4& 4& 0\\
0& 0& 0& 0& 0& 0& 0& 4& 0& 4& 4& 0& 0\\
0& 0& 0& 0& 0& 0& 0& 0& 4& 4& 0& 4& 0
\end{array}\right).\end{equation}
\textcolor{black}{The construction starts with two copies of the 6-qubit state whose check matrix is given in Eq.~\eqref{eq:checkmatrixV1}. The first two qubits of each state correspond to the logical legs of the building-block tensor. As illustrated in Fig.~\ref{fig:steane}, we contract the corresponding logical legs of these two states.} \textcolor{black}{To be more explicit, we begin with a 12-qubit system containing the two copies. We then apply Algorithm~\ref{algo:1} twice to perform the tensor contraction, once for each pair of logical legs. This process results in an 8-qubit state.} \textcolor{black}{Finally, we interpret this 8-qubit state as an encoding map for a 7-qubit code. We designate one of the eight legs as the logical qubit and the remaining seven as the physical qubits. The stabilizer group for this 7-qubit code is given by the check matrix:}
\begin{equation}
\scalebox{0.8}{$
\left(
\begin{array}{ccccccc|ccccccc|c}
1& 0& 1& 0& 1& 0& 1& 0& 0& 0& 0& 0& 0& 0& 0\\
0& 1& 1& 0& 0& 1& 1& 0& 3& 4& 0& 0& 3& 4& 2\\
0& 0& 0& 1& 1& 1& 1& 0& 0& 0& 0& 0& 7& 0& 1\\
0& 0& 0& 0& 0& 0& 0& 4& 0& 4& 0& 4& 0& 4& 0\\
0& 0& 0& 0& 0& 0& 0& 0& 4& 4& 0& 0& 4& 4& 0\\
0& 0& 0& 0& 0& 0& 0& 0& 0& 0& 4& 4& 4& 4& 0
\end{array}
\right).$}
\end{equation}

{\color{black}We can verify the distance of the code by computing its weight enumerator polynomials \cite{Shor1997}. For any quantum code over qubits whose projection onto the code subspace is $\Pi$, they are given by a pair of polynomials

\begin{align}
    A(z) &= \sum_{d=0}^nA_dz^d = \sum_{E\in \mathcal P^n }\Tr[E\Pi]\Tr[E\Pi]z^{wt(E)}\\
    B(z) &= \sum_{d=0}^nB_dz^d = \sum_{E\in \mathcal P^n }\Tr[E\Pi E\Pi]z^{wt(E)}
\end{align}
where $\mathcal P^n$ is the Pauli group over $n$ qubits and $wt(E)$ computes the weight of the Pauli string $E$.  The two polynomials are related to each other by the MacWilliams identity and encodes key information of a quantum code or state. See \cite{Cao2023} and references therein for details. For Pauli stabilizer codes, the coefficients $A_d, B_d$ also correspond to the number stabilizers or normalizers of weight $d$. 

Here we make use of two properties of enumerators which hold for any quantum code. Firstly, the enumerators are invariant under local unitary transformations $\Pi \rightarrow U\Pi U^{\dagger}$ where $U$ acts as a tensor product over all the physical qubits. Therefore, two codes must not be LU equivalent if they have different enumerator polynomials. However, the converse is not true --- there can be codes or states with the same enumerator but are not LU-equivalent. For example, the enumerators of a $GHZ_4$ state are identical to those of two Bell states.  Secondly, Ref.~\cite{Ashikhmin1999} shows that for any quantum code, its distance is the smallest $d$ for which $B_d-A_d>0$.
}

The weight enumerators for the resulting code are
\begin{align}
    A(z)=1+&21z^4+42z^6\\
    B(z)=1+21z^3+&21z^4+126z^5+42z^6+45z^7,
\end{align}
which are the same as those of the Steane code.
Although its codewords $|\bar{0}\rangle,|\bar{1}\rangle$ are not Pauli stabilizer states, \textcolor{black}{this code is related to the Steane code by a local unitary (LU) transformation such that $U\mathcal{S}_{XP} U^{\dagger} = \mathcal{S}_{\rm Steane}$
where $U=P^{7/2}_2Z_4 P_6^{7/2}$, $P_i^{1/2}$ is defined as a unitary acting on the $i$th qubit $\begin{pmatrix}
    1&0\\
    0&\omega
\end{pmatrix}$, and $\mathcal{S}_{XP}, \mathcal{S}_{\rm Steane}$ are the stabilizer groups of the 7-qubit XP code and Steane code respectively. }

The second example is constructed out of 6-qubit tensors whose state description has a check matrix
\begin{equation}
\scalebox{0.8}{$\left(\begin{array}{cccccc|cccccc|c}
1& 0& 0& 1& 1& 1& 2& 0& 0& 2& 6& 6& 0\\
0& 1& 0& 1& 0& 1& 0& 0& 0& 2& 0& 6& 0\\
0& 0& 1& 1& 1& 0& 0& 0& 3& 6& 6& 4& 1\\
0& 0& 0& 0& 0& 0& 4& 0& 0& 4& 4& 4& 0\\
0& 0& 0& 0& 0& 0& 0& 4& 0& 4& 4& 0& 0\\
0& 0& 0& 0& 0& 0& 0& 0& 4& 4& 0& 4& 0
\end{array}\right).$}\end{equation}
\textcolor{black}{Applying the same contraction method as in the previous example yields a 7-qubit code with the stabilizer generators:}
\begin{equation}
\scalebox{0.8}{$
\left(\begin{array}{ccccccc|ccccccc|c}
 1&  0&  1&  0&  1&  0&  1&  0&  0&  2&  0&  0&  4&  6&  0\\
 0&  1&  1&  0&  0&  1&  1&  0&  3&  6&  0&  4&  3&  2&  2\\
 0&  0&  0&  1&  1&  1&  1&  0&  0&  0&  2&  0&  7&  2& 13\\
 0&  0&  0&  0&  0&  0&  0&  4&  0&  4&  0&  4&  0&  4&  0\\
 0&  0&  0&  0&  0&  0&  0&  0&  4&  4&  0&  0&  4&  4&  0\\
 0&  0&  0&  0&  0&  0&  0&  0&  0&  0&  4&  4&  4&  4&  0
\end{array}\right).$}\label{eq:checkmatrixV8}\end{equation}
In this case, the weight enumerators
\begin{align}
    A(z)=1+ &13z^4+24z^5+18z^6+8z^7\\
    B(z)=1+13z^3+ &53z^4+78z^5+74z^6+37z^7
\end{align}
are different from those of the Steane code, meaning that they are not LU-equivalent. 
Also note that the number of non-trivial logical operators in this code has fewer representations, meaning that it is less likely to suffer an undetectable logical error. 

In general, XP codes from random contractions would admit transversal non-Clifford gates. 
\textcolor{black}{We can generate a code with a transversal non-Clifford gate by re-designating a physical qubit as logical by Prop.~\ref{prop:3.1}. Here we transform the second example (Eq.~\eqref{eq:checkmatrixV8}) into a $[[6,2,2]]$ code by making the second physical qubit a new logical qubit. From the check matrix, it can be seen that the encoding map is indeed an isometry.}
\textcolor{black}{
First, we move the columns of the second physical qubit to the front and bring the matrix to its canonical form. Next, we remove all rows that have non-trivial support on this qubit. Finally, we delete the columns corresponding to this qubit. The resulting check matrix has four rows, and the code features one non-diagonal logical operator for each logical qubit. According to Thm.~\eqref{counting theorem}, these properties confirm that this is a valid XP code.}

\textcolor{black}{Under this change, a stabilizer from the original code that acts on the second qubit becomes a logical operator for the new code. Consider the stabilizer given by the second row of the check matrix in Eq. (V.8):
\begin{equation}
    S = \omega^2(X_2 P_2^3) (X_3 P_3^6) P_5^4 (X_6 P_6^3) (X_7 P_7^2)
\end{equation}
This operator $S$ is now interpreted as a logical operator for the $[[6,2,2]]$ code. By Prop.~\ref{prop:3.1}, its action is separated into a logical part and a physical part; it acts as the identity on the first logical qubit ($L_1$), as $(X P^3)^t$ on the new logical qubit ($L_2$), and as a transversal gate on the remaining six physical qubits.}

\textcolor{black}{The logical component is therefore $\overline{I_{L_1} \otimes (P^3X)_{L_2}}=(X_3 P_3^6) P_5^4 (X_6 P_6^3) (X_7 P_7^2)$. This is a non-Clifford logical operation, since for $N=8$ the $P^3$ gate is a non-Clifford phase gate. The physical part of the operator is transversal by construction. Note that the ``code shortening'' in Prop~\ref{prop:3.1} that converts non-Pauli stabilizer into transversal non-Clifford logical gates is very general and applies to other types of unitary gates regardless of the level of the Clifford hierarchy.}



On the other hand, identifying XP codes that are (LU-equivalents of) Pauli stabilizer states can also be useful, especially in the context of magic state distillation. For general stabilizer codes, we lack an effective method to identify its transversal non-Clifford gates, which may exist as a symmetry of the code, but is not apparent in the PSF. As QL can identify such codes more effectively, it can also be applied to tabulate ``hidden symmetries'' of potentially well-known Pauli stabilizer codes.

In addition to random contraction, an obvious method to generate other novel XP codes for which the code properties can be efficiently computed is to use such small XP codes as seed codes, and fuse them into a tensor network that is efficiently contractible. Some obvious choices are tree tensor networks, which correspond to code concatenation, or 2D hyperbolic tensor networks, which generate non-stabilizer versions of holographic codes. For instance, by replacing the seed tensors in~\cite{holosteane} with our XP codes above, one can build up non-Pauli stabilizer versions of holographic Steane codes (Fig.~\ref{fig:holoxp}). Although it was shown by~\cite{Cao:2023mzo} that Pauli stabilizer codes cannot support non-trivial area operators, XP extensions of holographic codes may circumvent the no-go theorem and produce helpful toy models that contain the necessary ingredients for emergent gravitational backreactions.
\begin{figure}
    \centering
\includegraphics[width=0.3\textwidth]{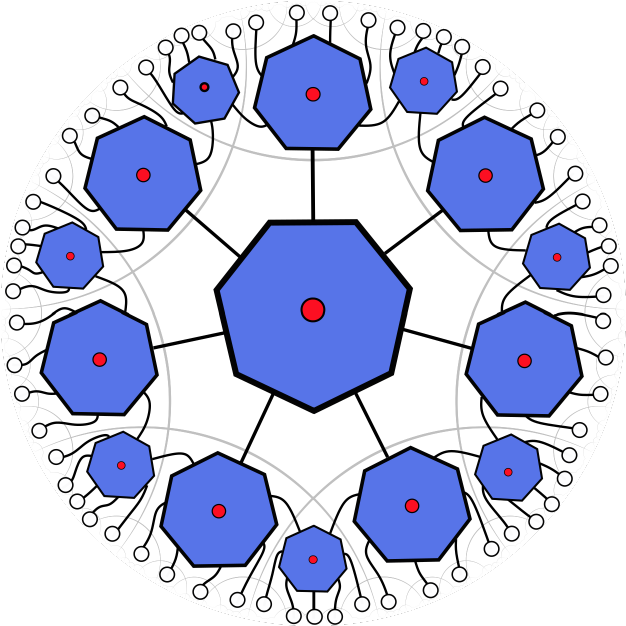}
    \caption{A holographic XP code where each heptagon represents the $[[7,1,3]]$ XP Lego that is not LU-equivalent to the Steane code. Red dots mark the logical legs.}
    \label{fig:holoxp}
\end{figure}

Using QL, it is easy to read off single-qubit transversal non-Clifford gates supported on the larger code through operator matching. Some examples were discussed in~\cite{Cao2022} by gluing different atomic codes with transversal $T$ gates. Here we do something slightly different --- we generate even smaller atomic codes that support fault-tolerant non-Clifford gates by tracing legs on the same code. 
Our starting point is the $[[15,1,3]]$ Reed-Muller code, a CSS code that can be equivalently seen as an $N=2$ XP code. In other words, this code space can be exclusively stabilized by Pauli stabilizers and its self-trace is a CSS code. However, the presence of logical $T$ in this language is hidden. On the other hand, this code can also be interpreted as an $N=8$ XP code, allowing us to better discern the presence of the transversal $T$ gate during operator matching. 
As $T$ is matched to $T^{\dagger}$, one method of self-trace that guarantees transversal $T$ from operator matching is to trace two physical legs with an additional insertion of $X$ operator (Fig.~\ref{fig:RMtrace}) by noting that $XT=\gamma T^{\dagger}X$ where $\gamma$ is a global phase. 
{\color{black} 
Note that this self-tracing with $X$ insertion is also sufficient to preserve transversality of codes that support bitwise transversal phase gates $$P(\varphi) = \begin{pmatrix}
   1 & 0\\
   0 & e^{i\varphi}
\end{pmatrix}$$
with any $\varphi$, where the matching condition $XP(\varphi)X= e^{i\varphi}P(\varphi)^{\dagger}$ is always satisfied. Hence such procedures can be applied to codes with transversal other phase gates in the Clifford hierarchy\footnote{We thank the anonymous referee for pointing this out.}.
}

Although the space of small triorthogonal codes is already heavily constrained~\cite{CalderbankT,Kay,Nezami}, smaller codes with fault-tolerant $T$ gates are possible~\cite{vasmer2022morphing} by breaking transversality weakly. Indeed, these self-traces typically yield codes with $Z$-distance $d_Z=1$ even though $d_X\geq 2$. 
For $n=13,11,9,7$, we can easily find codes with random self-contractions where the offending weight one logical $\bar{Z}$ is unique. Therefore, one can raise the Z distance by concatenating the support of weight-1 logical operator with another code $\mathcal{C}_Z$ that has $d_Z>1$. The resulting code has size at least $n+1$ and a fault-tolerant T gate $\bar{T}=T^{\otimes n-1}\otimes K$, similar to the $n=10$ code by~\cite{vasmer2022morphing}. The form of $K$ depends on the choice of $\mathcal{C}_Z$, which may be obtained through operator pushing. 

\begin{figure}
    \centering
    \includegraphics[width=0.45\linewidth]{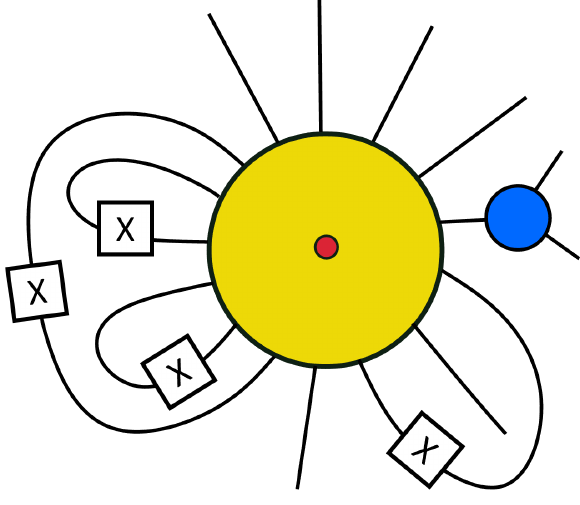}
    \caption{A $[[7,1,1]]$ code is produced from self-tracing the $[[15,1,3]]$ Reed-Muller code (yellow). It is then contracted with a repetition code, or GHZ tensor/X-spider (blue) to produce a $[[8,1,2]]$ code. Note that additional $X$s are inserted to ensure transversal $T$. Red dot denotes the logical input. }
    \label{fig:RMtrace}
\end{figure}

For example, four self-contractions produce a $[[7,1,d_Z=1/d_X=3]]$ code with check matrix

\begin{equation}
    \scalebox{0.8}{$\left(\begin{array}{ccccccc|ccccccc|c}
1& 1& 0& 0& 0& 1& 1& 0& 0& 0& 0& 0& 0& 0& 0\\
0& 0& 0& 1& 1& 1& 1& 0& 0& 0& 0& 0& 0& 0& 0\\
0& 0& 0& 0& 0& 0& 0& 1& 0& 1& 0& 1& 0& 1& 2\\
0& 0& 0& 0& 0& 0& 0& 0& 1& 1& 0& 1& 0& 1& 0\\
0& 0& 0& 0& 0& 0& 0& 0& 0& 0& 1& 1& 0& 0& 2\\
0& 0& 0& 0& 0& 0& 0& 0& 0& 0& 0& 0& 1& 1& 2
    \end{array}\right).$}
\end{equation}
Here we have set $N=2$ and the last column records the signs of the stabilizer generators. 
\textcolor{black}{A representation of the logical operators are $X_L = X_3X_6X_7$ and $Z_L = Z_3$.}
The code has weight enumerator polynomials:
\begin{align}
    A(z) &= 1+3z^2+23z^4+37z^6\\
    B(z) &= 1+z+3z^2+23z^3+23z^4\\
    &+111z^5+37z^6+57z^7.\nonumber
\end{align}
Hence other than the single weight-1 $\bar{Z}$, the rest of the logical operators all have weight $\geq 3$. {\color{black} Although this code only corrects a bit-flip error, its tensor is still interesting as it may be used as an non-isometric code to reduce stabilizer check weights while preserving the transversality of $T$ gates \cite{Cao:2025oep}. 

}

{\color{black}It is also straightforward to make it a error detection code through concatenation.}
If we choose $\mathcal{C}_Z$ to be a 2-qubit repetition code with stabilizer group $\langle XX\rangle$, then
\begin{align}
    K&\propto\cos(\pi/8)I\otimes I-i\sin(\pi/8) Z\otimes Z\\\nonumber
    &\propto diag(1, e^{i\pi/4}, e^{i\pi/4}, 1),
\end{align}
where we have dropped the global factors for simplicity. Then we arrive at an 8-qubit code, and its check matrix is
\setlength{\arraycolsep}{3.5pt}
\begin{equation}
    \left(\begin{array}{cccccccc|cccccccc|c}
1& 1& 0& 0& 1& 1& 0& 0& 0& 0& 0& 0& 0& 0& 0& 0& 0\\
0& 0& 1& 1& 1& 1& 0& 0& 0& 0& 0& 0& 0& 0& 0& 0& 0\\
0& 0& 0& 0& 0& 0& 1& 1& 0& 0& 0& 0& 0& 0& 0& 0& 0\\
0& 0& 0& 0& 0& 0& 0& 0& 1& 0& 0& 1& 0& 1& 1& 1& 2\\
0& 0& 0& 0& 0& 0& 0& 0& 0& 1& 0& 1& 0& 1& 1& 1& 0\\
0& 0& 0& 0& 0& 0& 0& 0& 0& 0& 1& 1& 0& 0& 0& 0& 2\\
0& 0& 0& 0& 0& 0& 0& 0& 0& 0& 0& 0& 1& 1& 0& 0& 2
    \end{array}\right)
\end{equation}
\setlength{\arraycolsep}{5pt}

with weight enumerators
\begin{equation}
\begin{aligned}
    A(z) &= 1+4z^2+18z^4+16z^5+28z^6+48z^7+13z^8\\
    B(z) &= 1+6z^2+20z^3+36z^4+120z^5+130z^6\\
    &\quad+116z^7+83z^8
\end{aligned}
\end{equation}
thus yielding an $[[8,1,2]]$ code with $\bar{T}=T^{\otimes 6}\otimes K$.

{\color{black}The logical $\bar{T}$ gate is fault-tolerant in the sense that it {propagates detectable errors to detectable errors}, which is the same notion of fault tolerance studied in \cite{vasmer2022morphing}. To see that the gate is FT, the only point of concern is the support on which $K$ acts as single qubit $T$ gate does not spread errors. Any $Z$ error will be caught by the inner code syndrome measurement. 

However, one can show that $X$ error can propagate to a logical phase gate times a detectable $X$ error on the outer code, which is detectable with a round of syndrome extraction. 
More explicitly, 
\textcolor{black}{an initial $X$ error on data qubits $j=7$ or $j=8$ transforms under the gate $K$ according to the identity:
$K X_j K^\dagger = \frac{1}{\sqrt{2}} (X_j \otimes I)(I \otimes I + iZ \otimes Z)$.
The key observation is that the propagated error contains an $X_i$ Pauli operator. Since this $X_j$ component anticommutes with any $Z$-type stabilizer acting on qubit $j$, the error is guaranteed to flip the corresponding syndrome bit. Therefore, any single qubit error remains detectable.}

As such, it is conceivable that the 8-qubit code can be used for magic state distillation where one always post-selects on the trivial syndrome. 
Even though this is similar to the $[[10,1,2]]$ protocol based on code morphing \cite{vasmer2022morphing}, there is no existing method for distilling the resource state for implementing the $K$ gate. Below we provide a schematic for FT synthesis of $|K\rangle$ using $|T\rangle$ and $|CS\rangle$ as a resource, but we leave its efficient synthesis to future work. 
Let 
\begin{align}
    &\ket{T} = T\ket{+}\\
    &\ket{CS} = CS\ket{++},
\end{align} allowing for the detection and correction of potential errors in these states. This construction is analogous to the approach detailed in Ref.~\cite{vasmer2022morphing}. Moreover, these resource states can be distilled following the procedures outlined in Ref.~\cite{Haah2018}, which predominantly yield only single-qubit $Z$ errors. 
The explicit circuit for magic state injection and the corresponding construction of the $K$ gate are illustrated in Fig.~\ref{fig:T-gate},\ref{fig:CS-gate}, and \ref{fig:K-gate}.

\begin{figure}
    \centering
    \scalebox{1.5}{$
\Qcircuit @C=1em @R=1em {
  & & \ctrl{1} & \qw    & \gate{S} & \qw \\
  \ket{T}& & \targ   & \qw & \meter \cwx 
}$}
    \caption{Implementation of $T$ gate through magic state injection.}
    \label{fig:T-gate}
\end{figure}
\begin{figure}
    \centering
    \scalebox{1.2}{
    $\Qcircuit @C=1em @R=1em {
    & \ctrl{2} & \qw      & \qw        & \qw      & \gate{Z} & \qw & \qw\\
    & \qw      & \ctrl{2} & \qw       & \qw      & \qw \cwx      & \gate{Z} & \qw \\
    \lstick{} & \targ    & \qw      & \gate{H}  & \qw & \meter \cwx   & \cwx \\
    \lstick{} & \qw      & \targ    & \gate{H}         &\qw &   \qw & \meter   \cwx 
      \inputgroupv{3}{4}{0.8em}{0.8em}{\ket{CS}}\\
}$}
    \caption{Implementation of $CS$ gate through magic state injection.}
    \label{fig:CS-gate}
\end{figure}

\begin{figure}
    \centering
    \scalebox{1.5}{$
\Qcircuit @C=1em @R=1em {
   & \qw & \gate{T}  &\qw  & \ctrl{1} &\qw & \qw \\
   & \qw   & \gate{T} &\gate{X} & \gate{S} & \gate{X} & \qw
}
$}
    \caption{A circuit construction of the $K$ gate.}
    \label{fig:K-gate}
\end{figure}

}

{\color{black}
Finally, we emphasize that the key difference from the 10-qubit protocol is that in order to achieve error suppression on the distillation output, the probability of physical $Z$ error on the input resource $|K\rangle$ need to be suppressed by the its weight. That is, weight-1 $Z$ errors occur with probability $O(p)$ while $ZZ$ error occur with probability $O(p^2)$. Whereas the $[[10,1,2]]$ protocol allows the $|CCZ\rangle$ resource to have Z error of any weight to have $O(p)$ physical error rate while also achieving overall logical error suppression on the distilled $|T\rangle$.

If we assume that requisite $|K\rangle$ state can be prepared where physical $Z$ of weight 1 occurs independently with probability $O(p)$, then it can be used for magic state distillation.
}

\section{Conclusion}\label{Section6}
Although QL accommodates all quantum codes in principle, its practical implementation beyond PSF has not been discussed. In this work, we provide the first systematic study of QL as applied to XP stabilizer codes. Using these XP codes as Lego blocks, we study how their symmetries transform under fusion or conjoining operations. For general XP regular codes/states, we provided a novel condition by which they can be identified. When using XP codes as Quantum Lego blocks, we find that XP codes are not closed under tracing (or Bell fusion) unlike their Pauli stabilizer cousins. However, if the fused tensor is XP, then operator matching remains sufficient in identifying its full set of XP symmetries. We provide atomic Legos that generate XP states and extend the check matrix conjoining algorithm to XPF, thus providing an efficient method to track the XP symmetries, similar to their PS counterpart in the previous work. Because the sequence of conjoining operations includes the application of quantum gates and measurements that are dual to XP states, these operations also provide an analog of~\cite{AaronsonGottesman} towards a more general classical simulation of regular XP processes.
In addition to the progress in formalism on both XP and QL, we also leverage the recent enumerator technique to construct optimal decoders, which addresses all i.i.d. single qubit error channels. Finally, we discover new instances of XP codes with better rates and distances from Lego fusion and discuss how they may be used to produce better toy models for AdS/CFT. We also identified a class of codes with fault-tolerant $T$ gates, including an $[[8,1,2]]$ code, which is the smallest such instance to the best of our knowledge.

As we have commented in the previous sections, various extensions and applications of the formalism and algorithms developed here have potential impact in a number of areas. 

\textit{New XP codes:} Although we made the first steps in generating XP codes in a more scalable way, we still lack a systematic search strategy to create novel XP codes with properties most relevant to fault tolerance. Future work in this area may also aim to generate candidates of self-correcting quantum memories as the code admits non-Abelian symmetries. 

\textit{Efficient classical simulations:} By improving and restricting the conjoining operation to act on certain XP processes that preserve the set of XP states, it is conceivable that a more complete extension of~\cite{AaronsonGottesman} can then be used to simulate non-(Pauli)-stabilizer states as well as certain unitary actions of higher Clifford hierarchies. Such an extension would also generalize~\cite{Bermejovega2013} as the states in question involve non-abelian symmetries.

\textit{Hidden symmetries:} Since XPF can be more effective in identifying transversal non-Clifford gates and hidden symmetries of Pauli stabilizer codes, techniques coupled with QL can also be used to identify codes that may be useful for efficient magic state distillations or protocols for code switching. 

\textit{Fault tolerance:} Finally, it remains unclear whether there are additional benefits by going beyond PSF. For example, one might hope that XP stabilizer codes are more advantageous compared to Pauli stabilizer codes in the practical context of fault-tolerant quantum computation. This has yet to be demonstrated, but would be a smoking gun evidence needed to justify a more extensive study of XP codes and non-(Pauli)-stabilizer codes at large. With our construction of explicit decoders, we can explore this angle in future work.

\section*{Acknowledgement}
We thank Xie Chen, Arpit Dua, Chris Pattison, John Preskill, Xiaoliang Qi, Thomas Scruby, Shuo Yang, and Haimeng Zhao for helpful discussions and comments.  We thank Benson Way for helpful comments and initial participation in the project. C.C. acknowledges the support by the Commonwealth Cyber Initiative at Virginia Tech, the Air Force Office of Scientific Research (FA9550-19-1-0360), and the National Science Foundation (PHY-1733907). The Institute for Quantum Information and Matter is an NSF Physics Frontiers Center.  Y.W. acknowledges the support of Shuimu Tsinghua Scholar Program of Tsinghua University.
R.S. acknowledges the support of Tsinghua University and Caltech through the Summer Undergraduate Research Fellowships (SURF) program.

\appendix
\section{Properties of XP code}
\label{app:xpproperty}
\subsection{Proof of Proposition~\ref{prop:2.1}}
In this section, we set up the necessary lemmas and produce the proof for proposition and theorems in Sec.~\ref{Section2}.

\begin{definition}[Group projector]\label{def:groupprojector}
    Let $G$ be a finite group of order $N_G$ that carries a unitary representation in a Hilbert space $\CH$. Then $\Pi_{G}\equiv \frac{1}{N_G}\sum_{g_i\in G}g_i$ is a projection operator in $\CH$.  We call $\Pi_G$ the group projector. Here the sum is in the sense of operators in a given Hilbert space.
\end{definition}

\begin{definition}
    Let $G$ be the same finite group as defined in Def \ref{def:groupprojector}. $\Pi_{\CH_G}$ is defined to be the projector to the maximal common eigenvalue $1$ eigenspace for all $g_i\in G$. This maximal common eigenspace is denoted $\CH_G$. With abuse of notation, $\Pi_{\CH_G}$ is called group code projector.
\end{definition}

\begin{lemma}[Group projector equals to group code projector]\label{lem:projector}
    Let $G$ be a finite group carrying a unitary representation in a Hilbert space $\CH$. Then $\Pi_G=\Pi_{\CH_G}$.
\end{lemma}
\begin{proof}
    On the one hand, by definition,  $\forall\ket{\psi}\in\CH_G$, $\Pi_G\ket{\psi}=\ket{\psi}$, so $\Pi_G \Pi_{\CH_G}=\Pi_{\CH_G}\Pi_G=\Pi_{\CH_G}$. This implies $\Pi_{\CH_G}\leq\Pi_G$. On the other hand, suppose $\Pi_{\CH_G}<\Pi_G$, then there must exist at least one state $\ket{\theta}\notin\CH_G$, such that $\Pi_{\ket{\theta}}\equiv\ket{\theta}\bra{\theta}$ satisfies $\Pi_{\ket{\theta}}<\Pi_G$. So $\Pi_{\ket{\theta}}\Pi_G\Pi_{\ket{\theta}}=\Pi_{\ket{\theta}}$. This is $\ket{\theta}\bra{\theta}\Pi_G\ket{\theta}\bra{\theta}=\ket{\theta}\bra{\theta}$, which requires $\bra{\theta}\Pi_G\ket{\theta}=1$. Because $g_i$ is a unitary, we have $\forall g_i\in G, \bra{\theta}g_i\ket{\theta}\leq 1$. To satisfy $\bra{\theta}\Pi_G\ket{\theta}=1$, it is necessary that $\forall g_i\in G, \bra{\theta}g_i\ket{\theta}= 1$. This implies $\forall g_i\in G, g_i\ket{\theta}=\ket{\theta}$. This contradicts with $\ket{\theta}\notin\CH_G$. So $\Pi_G=\Pi_{\CH_G}$.
\end{proof}

\begin{lemma}[Projector to the code subspace]\label{lem:codeprojector} For an XP code specified by its stabilizer $\mathcal{S}$ and code subspace $\mathcal{C}$, the projector to the code subspace $\CH_{\mathcal{C}}$, denoted as $\Pi_{\mathcal{C}}$ is defined to be the group code projector of its stabilizer group. $\Pi_{\mathcal{C}}\equiv\Pi_{\CH_S}$. Denote the logical identity group of the code as $\mathcal{S}_{LID}$. Then $\Pi_{\mathcal{C}}=\Pi_S=\Pi_{\mathcal{S}_{LID}}$.
\end{lemma}
\begin{proof}
    The first equality is by Lemma \ref{lem:projector}. To see the second one, we first note that in general $S\subseteq\mathcal{S}_{LID}$, so $\Pi_{\CH_S}\geq\Pi_{\CH_{\mathcal{S}_{LID}}}$. This is because to find $\CH_{\mathcal{S}_{LID}}$ one needs to impose (potentially) more constraints on top of $\CH_S$. On the other hand, by the definition of $\mathcal{S}_{LID}$, $\mathcal{C}\subseteq\CH_{\mathcal{S}_{LID}}$, this means $\Pi_{\mathcal{C}}\leq\Pi_{\CH_{\mathcal{S}_{LID}}}$. By Lemma \ref{lem:projector}, all these projectors are equivalent. 
\end{proof}

\textbf{Proposition \ref{prop:2.1}} immediately follows from Lemma \ref{lem:codeprojector} --- we can write code projector in terms of the stabilizer generators of $\mathbf{S}_X$ and $\mathbf{S}_Z$. 
Because $\Pi_\mathcal{C}=\Pi_{S}$ by Lemma    \ref{lem:codeprojector}, one can use Property 1 in the Appendix B of~\cite{Webster2022} to write $\Pi_S$ explicitly as stated in Proposition~\ref{prop:2.1}. The degree (the minimal integer $m$ such that $S_{Z_j}^m=\mathbb{I}$) of a given $S_{Z_j}$ might be a proper divisor of $N$, but this does not change the result of the sum.

\subsection{Proof of Theorem~\ref{counting theorem}}
We first review some notations and definitions introduced by~\cite{Webster2022}. The code subspace of an XP code is the intersection of the $+1$ eigenspaces of diagonal operators and non-diagonal operators.
The diagonal operators determine the Z-support of the codewords.
We define the Z-support of a state as 
\begin{equation}
    \mathrm{ZSupp}(\ket{\psi})=\{\mathbf{e}\in \mathbb{Z}_2^n : \braket{\mathbf{e}|\psi}\neq0\}.
\end{equation}
Denote the codewords as $\ket{\kappa_i}$.
It has been proved that each codeword can be written as
\begin{equation}
    \ket{\kappa_i} = \mathbf{O}_{\mathbf{S}_X} \ket{\mathbf{m}_i}
\end{equation}
where 
$\mathbf{O}_{\mathbf{S}_X} = \sum_{S\in\braket{\mathbf{S}_X}} S$
and $\mathbf{m}_i \in \mathbb{Z}_2^n$. 
The set of orbit representatives $\{\mathbf{m}_i\}$ is denoted as $E_m$.
The Z-support of the codeword $\ket{\kappa_i}$ has a coset decomposition 

\begin{align}
    \mathrm{ZSupp}(\ket{\kappa_i}) &= \mathbf{m}_i + \langle S_X \rangle \\\nonumber
    &= \{(\mathbf{m}_i + \mathbf{u}S_X)\mod 2: \mathbf{u}\in\mathbb{Z}_2^r\}.
\end{align}
Here $S_X$ refers to the x-part of $\mathbf{S}_X$, a set of binary strings.
The Z-support for the entire codespace can then be represented as the union of the Z-supports of all the codewords
\begin{equation}
    E = \bigcup_i \mathrm{ZSupp}(\ket{\kappa_i}) = E_m + \langle S_X\rangle.
\end{equation}
$E_m$ can be further decomposed into the cosets of $L_X$, where $L_X$ is the x-part of the non-diagonal logical operators
\begin{equation}
    E_m = E_q + \braket{L_X}.
\end{equation}
We refer to $E_q$ as the \textit{core} of the code, with cardinality $|E_q|$.
XP codes can be categorized based on the size of $|E_q|$.
If $|E_q|=1$, then this code is called an \textit{XP-regular code}; otherwise, it is non-regular.

We now present a few novel results on XP regular states, and by extension regular XP codes through the Choi-Jamiolkowski isomorphism. Recall that the number of stabilizer generators $m$ is equal to the number of qubits $n$ for Pauli stabilizer states. We extend this equality to XP states and establish a connection between the number of canonical generators and $n$. This property helps us identify potential XP states.

An XP-regular code can be mapped to a CSS code through a unitary transformation given by:
\begin{equation}
U = \sum_{\mathbf{e}_{ij}\in E} \omega^{p_{ij}} |\mathbf{e}_{ij}\rangle\langle \mathbf{e}_{ij}| + \sum_{\mathbf{e}\in \mathbb{Z}_2^n \backslash E} |\mathbf{e}\rangle\langle \mathbf{e}|,
\end{equation}
where $E$ is the Z-support of the codespace. 
This unitary transformation reveals that the difference between an XP-regular code and its corresponding CSS code lies only in the phase factors in front of the Z-supports, while they both share the same Z-support. 
Consequently, there exists a set of diagonal Pauli operators $\mathbf{R}_Z$ that uniquely stabilize these Z-supports.

\begin{lemma}
For an XP-regular code, the group projector of its diagonal logical identity (LID) subgroup $\langle \mathbf{S}_Z \rangle$ is equivalent to the group projector of the diagonal LID subgroup $\langle \mathbf{R}_Z \rangle$ of its corresponding CSS code.
\end{lemma}
\begin{proof}
The group projector $\Pi_{\langle \mathbf{S}_Z \rangle}$ projects onto the Z-support of the XP code, while $\Pi_{\langle \mathbf{R}_Z \rangle}$ projects onto the Z-support of the corresponding CSS code. Since they share the same Z-support, the two projectors must be equal:
\begin{equation}\label{Z-support Projector}
\Pi_{Z} = \frac{1}{|\braket{\mathbf{S}_Z}|}\sum_{s\in \langle \mathbf{S}_Z \rangle} s = \frac{1}{2^{|\mathbf{R}_Z|}} \sum_{r \in \langle \mathbf{R}_Z \rangle} r.
\end{equation}
\end{proof}

Based on the previous lemma, we can now establish the following counting theorem, which establishes a relationship between the number of non-diagonal LID generators $\mathbf{S}_X$, the number of non-diagonal logical operator generator $\mathbf{L}_X$ of an XP-regular code, and the number of diagonal LID generators for its corresponding CSS code.
\begin{theorem}
    An $n$-qubit XP-regular code $\mathcal{C}$ must have $|\mathbf{S}_X|+|\mathbf{L}_X|+|\mathbf{R}_Z|=n$ or $|\mathbf{S}_X|+|\mathbf{R}_Z|=n$ for an XP regular state.
\end{theorem}

\begin{proof}
    Taking trace on the projector onto the codespace, we get
    \begin{align}
        \Tr(\Pi_{\mathcal{C}}) &= \dim \mathcal{C} = 2^{|\mathbf{L}_X|}\\
                 &= \Tr\left(\frac{1}{|\mathcal{S}|} \sum_{s\in \mathcal{S}} s\right)\\
                 &= \Tr\left(\frac{1}{2^{|\mathbf{S}_X|}} \frac{1}{|\braket{\mathbf{S}_Z}|} \sum_{s \in \braket{\mathbf{S}_Z}} s\right) \\
&= \Tr\left(\frac{1}{2^{|\mathbf{S}_X|}} \frac{1}{2^{|\mathbf{R}_Z|}} \sum_{s \in \braket{\mathbf{R}_Z}} s\right) \\
                 &= \frac{1}{2^{|\mathbf{S}_X|}} \frac{1}{2^{|\mathbf{R}_Z|}} 2^n
    \end{align} 
    The equal sign in the third row arises from the fact that for any group element $s \in \mathcal{S}$, if it contains a non-diagonal part, then its contribution becomes zero after tracing. 
    Consequently, we only need to consider the diagonal subgroup in the summation.
    The equal sign in the fourth row arises from the observation that within the diagonal Pauli subgroup $\langle \mathbf{R}_Z \rangle$, only the identity element has a non-zero trace.  We complete the proof by taking logarithm on both sides of the equation.
    A logical operator on a one-dimensional codespace must either be a stabilizer, or give a phase to the state.
    Consequently, for XP states, we have $|\mathbf{L}_X|=0$.

    An alternative proof attains the same result by applying the rank-nullity theorem for $E$ as a matrix, where its support maps to the generators of $\mathbf{S}_X$ while its kernel to the generators of $\mathbf{R}_Z$.
\end{proof}

The above theorem applies to all XP-regular codes. 
However, it is inconvenient to consider $\mathbf{S}_X$, which belongs to the LID, and $\mathbf{R}_Z$, which belongs to the stabilizers of the corresponding CSS code, together. 
A more natural approach is to count both $\mathbf{S}_X$ and $\mathbf{S}_Z$ simultaneously, which is equivalent to counting the number of generators of LID. 
In the following, we will prove that this more convenient counting theorem holds for $N=2^t$, which includes the states we are interested in (those capable of expressing the $S$ gate and $T$ gate naturally).

\begin{theorem}\label{SZ=RZ}
    If $N=2^t$, then $|\mathbf{S}_Z|=|\mathbf{R}_Z|$.
\end{theorem}
\begin{proof}
    Let $s=|\mathbf{S}_Z|$ and $r=|\mathbf{R}_Z|$.
    First, we prove that $s<r$ is not possible.
    Since $\braket{\mathbf{R}_Z} \subset \braket{\mathbf{S}_Z}$, any generator of the Pauli subgroup can be represented by the product of generators of the diagonal LID group
    \begin{equation}
        R_{Zi} = \prod_{j=1}^{s} S_{Zj}^{m_{ij}}
    \end{equation}
    where $m_{ij}\in \mathbb{Z}_N$ are entries of $M\in \mathbb{Z}_N^{r\times s}$.
    The group multiplication of $R_{Zi},R_{Zj}$ can be represented by the addition of corresponding rows in $M$
    \begin{align}
        R_{Zi}R_{Zj} 
        &= \left(\prod_{k=1}^{s} S_{Zk}^{m_{ik}}\right)
        \left(\prod_{k=1}^{s} S_{Zk}^{m_{jk}}\right)\\\notag
        &= \prod_{k=1}^{s} S_{Zk}^{m_{ik}+m_{jk}}.
    \end{align}
    Therefore, $\braket{\mathbf{R}_Z}$ can be represented by $\mathrm{Span}_{\mathbb{Z}_N}(M)$, which can be written as $\mathrm{Span}_{\mathbb{Z}_N}\left(\mathrm{How}_{\mathbb{Z}_N}(M)\right)$.
    Since the matrix $M$ has dimensions $r \times s$, and $s<r$, it follows that the Howell form $\mathrm{How}_{\mathbb{Z}_N}(M)$ has at most $s$ non-zero rows. This implies that $\braket{\mathbf{R}_Z}$ can be generated by only $s$ Pauli generators, which leads to a contradiction.

    Conversely, we can demonstrate that $s>r$ is also not possible by constructing a Pauli operator in the LID that does not belong to the Pauli subgroup.
    To begin, we change the order of the qubits so that $\mathbf{R}_Z$ can be transformed into canonical form, with the first $r$ columns serving as pivoting columns. This arrangement implies that any non-identity element in $\braket{\mathbf{R}_Z}$ must act on the first $r$ qubits. Similarly, we transform $\mathbf{S}_Z$ into canonical form, resulting in $s$ rows.
    Consequently, there must exist a row in $\mathbf{S}_Z$ that consists entirely of zeros on the first $r$ qubits. Let's denote the operator corresponding to this row as $g$. Since $N=2^t$, the order  of any element $g$ must be a power of $2$, which we denote as $x$. Now, $g^{x/2}$ must be a diagonal Pauli operator and hence belongs to $\braket{\mathbf{R}_Z}$.
    However, this operator $g^{x/2}$ has no support on the first $r$ qubits, leading to a contradiction. It implies that $g^{x/2}$ is an element of $\braket{\mathbf{R}_Z}$ that acts only on the remaining qubits beyond the first $r$ qubits. This is not possible since all non-identity elements in $\braket{\mathbf{R}_Z}$ must act on the first $r$ qubits.
    As a result, we have shown that $s>r$ is not a valid scenario, completing the proof.

    Therefore, $|\mathbf{S}_Z|=|\mathbf{R}_Z|$.
\end{proof}

\subsection{XP states are not dense}
\begin{lemma}\label{lemma:notdense}
   Denote the set of $n$-qubit XP states for a fixed precision $N$ as $\mathcal{C}_{\text{XP}_N}(n)$. Define 
   \[\mathcal{C}_{\text{XP}}(n)=\bigcup_{N\geq 2,N\in\mathbb{N}}\mathcal{C}_{\text{XP}_N}(n).\] $\mathcal{C}_{\text{XP}}(n)$ is not dense in the $n$ qubit Hilbert space $\mathcal{H}(n)$.
\end{lemma}
\begin{proof}

If $\mathcal{C}_{\text{XP}}(n)$ were dense in $\mathcal{H}(n)$, this means for any state $\ket{\psi}\in\mathcal{H}(n)$, and any $\epsilon>0$, there exists a state $\ket{\xi}\in \mathcal{C}_{\text{XP}}(n)$, such that $\|\ket{\psi}-\ket{\xi}\|<\epsilon$.

However, a generic XP state in $\mathcal{C}_{\text{XP}}(n)$ is of the form 
    \begin{equation}
        \ket{\psi_i}=\frac{1}{2^{|S_X|/2}}\sum_{1\leq j\leq 2^{|S_X|}} \omega^{p_{ij}}\ket{e_{ij}},
    \end{equation}
where $|S_X|$ is the number of the generator of the non-diagonal canonical stabilizer generators, $p_{ij}\in \mathbb{Z}_{2N}$ and $\omega^{p_{ij}}$ is a phase. $e_{ij}\in\mathbb{Z}_{2}^n$ represents a basis of computational bases.(Eq. 35 of~\cite{Webster2022})

Because of the constrained form of the XP states, $\mathcal{C}_{\text{XP}}(n)$ is not dense in $\mathcal{H}(n)$. To give a counterexample, let the state be $\ket{\psi}=\frac{1}{2}\ket{e_1}+\frac{\sqrt{3}}{2}\ket{e_2}$, then one can show that $\forall \xi\in \mathcal{C}_{\text{XP}}$, $\|\ket{\psi}-\ket{\xi}\|^2\geq 2-(\sqrt{6}+\sqrt{2})/2$. So $\ket{\psi}$ is not in the $\epsilon$-neighbourhood of any $XP$ states if we choose $\epsilon < 2-(\sqrt{6}+\sqrt{2})/2$.

\end{proof}

\section{Tracing XP codes}\label{app:tracingXP}
\subsection{Characterization of the Post-trace State}
\emph{Restricting Z-supports---}
Performing self-tracing on the first two qubits of an XP state can be understood as selecting the Z-supports that have $00$ or $11$ on the first two qubits and retaining their phase, while discarding those Z-supports with $01$ or $10$ on the first two qubits.
We can classify XP states into two cases according to this observation:
\begin{enumerate}
    \item The state exclusively contains Z-supports with $00$ and $11$ on the first two qubits.
    By employing $E_m = \mathbf{m} + \braket{S_X}$, it is evident that the only non-diagonal LID generator that acts on the first two qubits is $S_{X1}=XP_N(1,1,\mathbf{x}'|\mathbf{z}|p)$.
    \item The state contains Z-supports with all $00,01,10,11$ on the first two qubits.
    Then there are two non-diagonal LID generators acting on the first two qubits $S_{X1}=XP_N(1,0,\mathbf{x}'|\mathbf{z}|p)$, $S_{X2}=XP_N(0,1,\mathbf{x}^{''}|\mathbf{z}'|p')$.
\end{enumerate}
We favor the first case over the second case because the number of Z-supports remains unchanged before and after tracing
\begin{equation}
    \ket{\psi_1}\rightarrow\ket{\psi_{1t}},
\end{equation}
where $\ket{\psi_1}$ is a state in case 1 and $\ket{\psi_{1t}}$ its post-trace state.
In the second case, Z-supports with $01$ and $10$ on the first two qubits do not contribute to the traced state, but they introduce additional constraints on the stabilizers for this state.
Hence, we aim to find a method to transform the second case into the first case.

We can achieve this by manipulating the generators as follows:
\begin{itemize}
    \item Remove $S_{X1},S_{X2}$ from $\mathbf{S}_X$.
    \item Add $S_{X1}S_{X2}$ into $\mathbf{S}_X$.
    \item Add $XP_N(0,0,\mathbf{0}|1,N-1,\mathbf{0}|0)=P_1\otimes P^{N-1}_2$ into $\mathbf{S}_Z$.
    \item Add $Z_1\otimes Z_2$ into $\mathbf{R}_Z$.
\end{itemize}
By performing these manipulations, we create a new XP state with restricted Z-support in case 1. 
It can be easily verified that $Z_1\otimes Z_2$ was not a stabilizer in the initial state, as Z-supports with $01$ and $10$ on the first two qubits were present. 
The validity of Theorem~\ref{counting theorem} still holds for this modified state.
According to Theorem~\ref{SZ=RZ}, we know that adding $P_1\otimes P_2^{N-1}$ is sufficient to generate the new diagonal LID group.
The restricted state produces the same post-trace state as the original state.
Therefore, we have the tracing procedure
\begin{equation}
    \ket{\psi_2}\rightarrow\ket{\psi_{2r}}\rightarrow\ket{\psi_{2t}}
\end{equation}
where $\ket{\psi_2}$ is a state in case 2, $\ket{\psi_{2r}}$ is its restricted state, and $\ket{\psi_{2t}}$ is the post-trace state of  $\ket{\psi_{2r}}$.
Since $\ket{\psi_2}$ and  $\ket{\psi_{2r}}$ generate the same post-trace state, we can study  $\ket{\psi_{2r}}$, which is in case 1, instead of $\ket{\psi_{2}}$.
From this point forward, we will solely focus on states in case 1 since any state in case 2 can be transformed into case 1 using this method. 
This simplifies our analysis.

\par
\subsection{Proof of Theorem~\ref{Operator matching is enough}.}
\begin{proof}
Let's start with the diagonal LID operator on the post-traced state in the form $XP_N(\mathbf{0}|\mathbf{z}|p)$. We can extend this operator back to act on the pre-traced state as follows:
\begin{equation}
XP_N(\mathbf{0}|\mathbf{z}|p) \rightarrow XP_N(0,0,\mathbf{0}|a,N-a,\mathbf{z}|p),
\end{equation}
where $a\in \mathbb{Z}_N$ is arbitrarily chosen.
We can easily verify that the extended operator is a logical identity on the pre-traced state. 
For any $\mathbf{e}$ and $\mathbf{f}$ such that $\ket{00\mathbf{e}}$ and $\ket{11\mathbf{f}}$ belong to the Z-support $E$,
\begin{equation}\begin{aligned}
    &XP_N(0,0,\mathbf{0}|a,N-a,\mathbf{z}|p) \ket{00\mathbf{e}}\\ 
    = &XP_N(0,0|1,N-1|0)\ket{00} \otimes XP_N(\mathbf{0}|\mathbf{z}|p)\ket{\mathbf{e}}\\
    = &\ket{00\mathbf{e}}\\
    &XP_N(0,0,\mathbf{0}|a,N-a,\mathbf{z}|p) \ket{11\mathbf{f}}\\ 
    = &XP_N(0,0|1,N-1|0)\ket{11} \otimes XP_N(\mathbf{0}|\mathbf{z}|p)\ket{\mathbf{f}}\\
    = &\ket{11\mathbf{f}}.
\end{aligned}\end{equation}
Furthermore, note that $XP_N(0,0,\mathbf{0}|a,N-a,\mathbf{z}|p)$ is a matched operator.
As a result, all diagonal LID operators on the post-traced state can be derived from operator matching.

Next, we explore the behavior of non-diagonal LIDs on Z-supports. 
The non-diagonal LID changes the Z-support and introduces an additional phase. 
To track the change in Z-support, we focus on the X-part of the operator. 
We consider a specific non-diagonal LID denoted as $XP_N(x_1, x_2, \mathbf{x}|z_1, z_2, \mathbf{z}|p)$. 
Its action on the Z-supports can be expressed as follows:
\begin{equation}\begin{aligned}
&XP_N(x_1, x_2, \mathbf{x}|z_1, z_2, \mathbf{z}|p)\ket{e_1, e_2, \mathbf{e}}\\
= &\omega^{p + 2(z_1e_1 + z_2e_2) + 2\mathbf{z}^T\mathbf{e}} \ket{e_1\oplus x_1, e_2\oplus x_2, \mathbf{e}\oplus \mathbf{x}}.
\end{aligned}\end{equation}
Here, $e_1 = e_2$, $x_1 = x_2$, and $(e_1,e_2,\mathbf{e})\in E$.
Now, if the post-trace state is an XP state, it implies the existence of a new non-diagonal LID with the form $XP_N(\mathbf{x}|\mathbf{z}'|p')$ so that the new Z-support can be fully generated. 
The action of this new operator on codewords is given by:
\begin{equation}\begin{aligned}
&XP_N(\mathbf{x}|\mathbf{z}'|p') \ket{\mathbf{e}}
= \omega^{p'+2\mathbf{z'}^{T}\mathbf{e}} \ket{\mathbf{e}\oplus \mathbf{x}}\\
= &\omega^{p + 2(z_1e_1 + z_2e_2) + 2\mathbf{z}^T\mathbf{e}} \ket{\mathbf{e}\oplus \mathbf{x}}.
\end{aligned}\end{equation}
The invariance of the relative phase between Z-supports before and after tracing allows us to extend the operator as follows:
\begin{equation}
XP_N(\mathbf{x}|\mathbf{z}'|p') \rightarrow XP_N(x_1, x_2, \mathbf{x}|a, N-a, \mathbf{z}'|p')
\end{equation}
where $a\in \mathbb{Z}_N$. 
Consequently, the extended operator has the following action on the old Z-supports:
\begin{equation}\begin{aligned}
&XP_N(x_1, x_2, \mathbf{x}|a, N-a, \mathbf{z}'|p') \ket{e_1, e_2, \mathbf{e}}\\
&= \omega^{p'+2(ae_1+(N-a)e_2)+2\mathbf{z}^{'T}\mathbf{e}} \ket{e_1\oplus x_1, e_2\oplus x_2, \mathbf{e}\oplus \mathbf{x}}\\
&= \omega^{p + 2(z_1e_1 + z_2e_2) + 2\mathbf{z}^T\mathbf{e}} \ket{e_1\oplus x_1, e_2\oplus x_2, \mathbf{e}\oplus \mathbf{x}}.
\end{aligned}\end{equation}
This result confirms that the extended operator acts as an old generator on the old Z-supports. Therefore, we arrive at the theorem based on these findings.
\end{proof}

Tracing XP states constitutes a universal process However, it is important to note that XP states are not dense within the Hilbert space (Lemman \ref{lemma:notdense}). 
While tracing some XP states does yield another XP state, this cannot be always true in all the cases. 
Our objective is to discern the specific conditions under which tracing an XP state produces another XP state.
The operator matching rule, as well as the counting theorem, tells us that the number of matched generators must decrease by $2$.

\begin{lemma}
    The projector onto an XP state is equal to the average of the elements in the LID of this state.
\end{lemma}
\begin{proof}
    The group projector is equal to the code projector.
\end{proof}
\begin{remark}
    The preceding lemma might appear self-evident. 
    However, it's important to stress a crucial insight that follows.
    Let's consider the scenario where we trace an XP state to yield another XP state. 
    In this context, it's possible to derive the projector onto the post-trace state, denoted as $\Pi'$, from the projector onto the pre-traced state, denoted as $\Pi$. 
    Mathematically, this relationship is expressed as follows:
    \begin{equation}
    \Pi' = \Tr_{12}(\Pi) = \sum_{h\,\mathrm{match}} c_h h + \sum_{h'\,\mathrm{unmatch}} c_{h'} h',
    \end{equation}
Since $\Pi$ is a group projector of the logical identity group, it can be written as the sum of its group elements.
All its group elements can be classified into two categories: those that satisfy the operator matching rule on the first two qubits, and those not.
For those who satisfy the operator matching rule, denote their substring as $h$, and tracing produces a coefficient of exactly $c_h=1$.
Those who don't satisfy the operator-matching rule, denote their substring as $h'$, and they have coefficients $c_{h'}$ as complex numbers in general.
This contrasts with PSF, where $c_{h'}$ are bound to be zero.
It's worth noting that in general we cannot eliminate the second term from the unmatched operators which do not vanish. 
Nonetheless, the lemma provides a key insight: for XP states, the second summand must vanish such that the unmatched operators appear to mutually interfere and cancel with each other.
\end{remark}

When the second term does not vanish, however, tracing will produce non-XP states. For example, consider the check matrix for the 6-qubit state
\begin{equation}\left(\begin{array}{cccccc|cccccc|c}
1& 0& 0& 1& 1& 1& 2& 0& 0& 2& 6& 6& 0\\
0& 1& 0& 1& 0& 1& 0& 0& 0& 2& 0& 6& 0\\
0& 0& 1& 1& 1& 0& 0& 0& 3& 6& 6& 4& 1\\
0& 0& 0& 0& 0& 0& 4& 0& 0& 4& 4& 4& 0\\
0& 0& 0& 0& 0& 0& 0& 4& 0& 4& 4& 0& 0\\
0& 0& 0& 0& 0& 0& 0& 0& 4& 4& 0& 4& 0
\end{array}\right).\end{equation}
which we used to generate the $[[7,1,3]]$ code.
The state it stabilizes is
\begin{align}
        &\ket{000000}+\ket{100111}+\ket{010101}+\ket{110010}\\\nonumber
        +&\omega \ket{001110} + \omega \ket{101001} + \omega^5 \ket{011011} \\\nonumber
        &+ \omega^{13} \ket{111100}.
\end{align}
After tracing the first two qubits, the remaining state is 
\begin{equation}
        \ket{0000}+\ket{0010}+(\omega+\omega^{13})\ket{1110}.
\end{equation}
It is obvious that this is not an XP state, because the coefficients have different norms.
By operator matching, the resulting check matrix is 
\begin{equation}
    \left(\begin{array}{cccc|cccc|c}
1& 1& 1& 0& 3& 6& 6& 0& 1\\
0& 0& 0& 0& 4& 4& 0& 0& 0\\
0& 0& 0& 0& 0& 0& 0& 4& 0
    \end{array}\right).
\end{equation}
Only three generators remain for the LID group in the canonical form, which is fewer than the number of qubits. This indicates that the post-trace state is not an XP state.

\section{Decoder}\label{app:xpdecoder}

First, we prove some properties within the Z-supports of the codespace.
Let $\mathcal{S}_{\rm LID}$ be the logical identity (LID) group.
We define the projector $\Pi_Z=\Pi_{\mathbf{S_Z}}=\Pi_{\mathbf{R}_Z}$ to be the projector onto the Z-support.
\begin{lemma}
    $\forall r\in\braket{\mathbf{R}_Z}, \forall g\in \mathcal{S}_{\rm LID}$, $[r,g]=0$.
\end{lemma}
\begin{proof}
    The commutator can be simplified as follows: $rgr^{-1}g^{-1} = D_N(2 x_r z_g - 2 x_g z_r + 4x_r x_g z_r - 4 x_r x_g z_g)$.
    Considering that $x_r$ equals zero and each entry in $z_r$ is $N/2$, the anti-symmetric portion can only be $\pm I$. 
    Furthermore, since $\mathcal{S}_{\rm LID}$ forms a valid stabilizer group, it is essential to exclude $-I$.
    Hence, the commutator is constrained to solely yield $I$, indicating that the elements $r$ and $g$ indeed commute.
\end{proof}
\begin{corollary}
    $[\Pi_Z,g]=0, \forall g\in \mathcal{S}_{\rm LID}$.
\end{corollary}
\begin{proof}
    $\Pi_Z=\frac{1}{2^{|\mathbf{R}_Z|}}\sum_{r\in\braket{\mathbf{R}_Z}} r$, and each $r$ commutes with $g$.
\end{proof}

\begin{lemma}
    $\forall s\in \braket{\mathbf{S}_Z}$, $\Pi_Z s=s\Pi_Z=\Pi_Z$.
\end{lemma}
\begin{proof}
    $\Pi_Z$ is the group projector of $\braket{\mathbf{S}_Z}$.
    $s$ permutes the elements in $\braket{\mathbf{S}_Z}$, hence the group projector remains invariant.
\end{proof}

\begin{lemma}
    $\forall g\in \mathcal{S}_{\rm LID}$, $\Pi_Z g$ is Hermitian.
\end{lemma}
\begin{proof}
    \begin{align}
        (\Pi_Z g)^{\dagger} = &g^{\dagger} \Pi_Z = g^{-1}\Pi_Z\\
        = &g g^{-1} g^{-1} \Pi_Z = gs\Pi_Z = g\Pi_Z = \Pi_Z g
    \end{align}
    for some $s\in\braket{\mathbf{S}_Z}$.
\end{proof}

\begin{lemma}
    $\forall g\in \mathcal{S}_{\rm LID}$, $(\Pi_Z g)^2 = \Pi_Z$.
\end{lemma}
\begin{proof}
    \begin{align}
        (\Pi_Z g)^2 = &\Pi_Z g\Pi_Z g = \Pi_Z gg\\
        = & \Pi_Z s = \Pi_Z
    \end{align}
    for some $s\in\braket{\mathbf{S}_Z}$.
\end{proof}
\begin{corollary}
    $\forall g\in \mathcal{S}_{\rm LID}$, the eigenvalue of $\Pi_Z g$ can only be $\pm1,0$.
\end{corollary}

\begin{lemma}
    $\forall g,h \in \mathcal{S}_{\rm LID}$, $\Pi_Z g,\Pi_Z h$ commute.
\end{lemma}
\begin{proof}
    \begin{align}
        \Pi_Z g \Pi_Z h = & \Pi_Z gh = \Pi_Z gh (g^{-1}h^{-1}hg)\\
        = &\Pi_Z (gh g^{-1}h^{-1})hg = \Pi_Z s hg \\
        = &\Pi_Z hg = \Pi_Z h \Pi_Z g
    \end{align}
    for some $s\in\braket{\mathbf{S}_Z}$.
\end{proof}
The same arguments also apply to $\Pi_{s_z}= E_{s_z}\Pi_Z E_{s_z}^{\dagger}$ and $\mathbf{S}_{X'} = E_{s_z} \mathbf{S}_X E_{s_z}^{\dagger}$, since they are conjugated under the same error operator $E_{s_z}$. It is helpful to recall the fact that the conjugation of XP operators, denoted as $A_1A_2A_1^{-1} = A_2D_N(2\mathbf{x}_1\mathbf{z}_2 + 2\mathbf{x}_2\mathbf{z}_1 - 4\mathbf{x}_1\mathbf{x}_2\mathbf{z}_1)$, remains an XP operator~\cite{Webster2022}. 
This implies that under some XP error $E_s$, the error subspace is the code subspace defined by yet another XP code with conjugated generators $\langle E_s\mathbf{S}_X E_s^{\dagger},E_s\mathbf{S}_Z E_s^{\dagger}\rangle$.

As described in Sec.~\ref{Section5}, the decoding process comes in two steps. For the first step, we identify and measure $\mathbf{R}_Z$. This obtains syndromes $s_z$, which allows us to identify any error $E_{s_z}$ that gives the same syndrome. Because $\mathbf{R}_Z$ defines a Pauli stabilizer code, the above process is well-understood. For the second step of syndrome extraction, we either perform $E_{s_z}^{\dagger}$ to restore the code back to the code subspace or measure a set of modified checks $\mathbf{S}_X' = E_{s_z}\mathbf{S}_X E_{s_z}^{\dagger}.$ 

Note that both $\mathbf{S}_X$ and $\mathbf{S}_X'$ are manifestly non-Hermitian operators; however, the above lemmas show that they behave like Hermitian operators with eigenvalues $\pm 1$ in the subspaces supported on $\Pi_Z$ and $\Pi_{s_z}$ respectively. Therefore, one can use the usual syndrome extraction circuit to measure these non-diagonal checks. 
The measurement circuit for check $g=g_1\otimes g_2\otimes\dots\otimes g_n$ is 
\begin{widetext}
    \begin{equation}
    (H\otimes I) \left(\frac{1}{\sqrt{2}}\ket{0}\bra{0}\otimes I +  \frac{1}{\sqrt{2}}\ket{1}\bra{1}\otimes g\right) (H\otimes I) \ket{0}\otimes \ket{\psi}
    = \frac{1}{\sqrt{2}} \ket{+}\otimes\ket{\psi} + \frac{1}{\sqrt{2}}\ket{-}\otimes g\ket{\psi}
\end{equation}
\end{widetext}

\begin{figure}[htbp]
\centering
\scalebox{1.5}{
$\Qcircuit @C=1em @R=1em {
    \lstick{\ket{0}}     & \gate{H} & \ctrl{1} & \gate{H} & \meter & \qw \\
    \lstick{\ket{\psi}}  & \qw      & \gate{g} & \qw      & \qw    & \qw
}$}
\caption{Syndrome measurement circuit. $\ket{\psi}$ is the state to be measured, and $\ket{0}$ is an ancilla qubit. The measurement is performed in Z-basis.}
\label{MeasurementCircuit}
\end{figure}
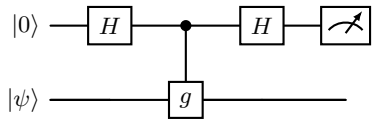
where $\ket{0}$ is an ancilla qubit and $\ket{\psi}$ is the state to be measured, as shown in Fig.~\ref{MeasurementCircuit}. 
Since we have already made the first round of measurement, then we know the Z-supports, so $\ket{\psi} = \Pi_Z\ket{\psi}$.
Since $g\Pi_Z$ is Hermitian, we can decompose $\ket{\psi}$ according to the eigenvalue of $g\Pi_Z$, say $\ket{\psi}=c_+\ket{\psi_+} + c_-\ket{\psi_-}$.
Then, the state after applying the circuit but before the measurement is
\begin{equation}
    c_+\ket{0}\ket{\psi_+} + c_-\ket{1}\ket{\psi_-}
\end{equation}
Therefore, measuring the first qubit in the Z-basis and getting $+1$ projects the physical qubits onto $\ket{\psi_+}$, and getting $-1$  projects onto $\ket{\psi_-}$.
This measurement procedure is valid because $g\Pi_Z$ is Hermitian, so its eigenvectors form a complete basis, and also $(g\Pi_Z)^2=\Pi_Z$ so that the eigenvalues of $g\Pi_Z$ can only be $\pm 1$ and $0$.

Let the syndromes we extract from the non-diagonal checks be $s_x$. 
To use the enumerator decoder, we now need to identify errors $E_{s_x}$ such that they have the same syndrome but consist only of diagonal operators so it commutes with $\Pi_Z$. We ask whether it is always possible. Let $\mathbf{p}$ denote the power of Pauli $Z$s in $E_{s_x}$ and the x half of the $i$th non-diagonal generator in the check matrix be $\mathbf{x}_i$. We then need 
$$ \mathbf{x}_i \cdot\mathbf{p}= \sigma_i,~~\mathrm{where}~\sigma_i\in\{0,1\}.$$ This is a system of $|\mathbf{S}_X|$ equations with $n$ unknowns. In the canonical form, we stack $\mathbf{x}_i$ to form a $|\mathbf{S}_X|\times n$ matrix, which has full row rank $r\leq n$ by definition. Therefore the system of equations should always admit a solution for any syndrome combination.

The most likely logical error $\bar{L}$ is then determined from the enumerator method by computing the logical error probabilities. For a straightforward application of the tensor network method introduced by~\cite{Cao2023}, we here only consider $\bar{L}$ that can be represented as a tensor product of single qubit operators in the physical Hilbert space. In Pauli stabilizer codes, all logical Pauli operators are transversal because Clifford unitaries map Pauli's to Pauli's. As logical Pauli operators form a nice unitary basis, it is sufficient to choose $\{\bar{L}\}$ in the Pauli basis. However, one needs to take extra care in the XP formalism because transversal logical gates supported by the code may not form a unitary basis in XP codes in general. For XP regular codes, one can obtain them from XP regular states. For $N=2^t$, we have shown that the rank-generator relation continues to hold. This allows us to identify transversal logical XP operators that form a unitary basis of logical operators.

\bibliographystyle{unsrtnat}
\bibliography{main_arxiv} 
\end{document}